\documentclass[onecolumn,10pt]{article}
\usepackage[top=2.54cm, bottom=2.54cm, left=2.54cm, right=2.54cm,a4paper]{geometry}

\usepackage[bitstream-charter, greekuppercase=upright,cal=cmcal]{mathdesign}
\usepackage[T1]{fontenc}

\usepackage[utf8]{inputenc} 
\usepackage{amsmath,amsthm}

\usepackage{mathabx}
\usepackage{verbatim}

\usepackage{xcolor}
\usepackage{mdframed}
\usepackage{units}

\definecolor{mblue}{rgb}{0.368417, 0.506779, 0.709798}
\definecolor{morange}{rgb}{0.880722, 0.611041, 0.142051}
\definecolor{mgreen}{rgb}{0.560181, 0.691569, 0.194885}
\definecolor{mred}{rgb}{0.922526, 0.385626, 0.209179}
\definecolor{mpurple}{rgb}{0.528488, 0.470624, 0.701351}

\usepackage[backend=biber,sorting=none,style=phys,biblabel=brackets,backref=false,bibencoding=utf8,maxnames=20,eprint=true]{biblatex}
\makeatletter
\pretocmd{\blx@head@bibintoc}{\phantomsection}{}{\ddt}
\makeatother
\DefineBibliographyStrings{english}{%
  backrefpage = {page},
  backrefpages = {pages},
}
\usepackage{titlesec}

\titleformat*{\section}{\bfseries}
\titleformat*{\subsection}{\normalsize\bfseries}
\titleformat*{\subsubsection}{\bfseries}
\titleformat*{\paragraph}{\large\bfseries}
\titleformat*{\subparagraph}{\large\bfseries}

\titlespacing\section{0pt}{12pt plus 4pt minus 2pt}{2pt plus 2pt minus 2pt}

\usepackage[bookmarks,colorlinks,breaklinks]{hyperref}  
\definecolor{dullmagenta}{rgb}{0.4,0,0.4}   
\definecolor{darkblue}{rgb}{0,0,0.4}
\hypersetup{linkcolor=red,citecolor=blue,filecolor=dullmagenta,urlcolor=darkblue} 

\newcommand{\ket}[1]{|#1\rangle}
\newcommand{\cket}[1]{|\widetilde{#1}\rangle}
\newcommand{\bra}[1]{\langle #1|}
\newcommand{\cbra}[1]{\langle \widetilde{#1}|}
\newcommand{\ketbra}[1]{\ket{#1}\bra{#1}}

\newcommand{\bracket}[2]{\langle #1|#2\rangle}
\def\tr{{\rm Tr}}

\def\eps{\varepsilon}

\def\cB{\mathcal B}

\def\cD{\mathcal D}
\def\cE{\mathcal E}
\def\cF{\mathcal F}

\def\cH{\mathcal H}

\def\cN{\mathcal N}

\def\cX{\mathcal X}

\def\cZ{\mathcal Z}

\renewcommand{\rho}{\varrho}
\renewcommand{\phi}{\varphi}
\def\id{\mathbb I}

\newcommand{\pg}{P}

\newcommand{\mix}{\mu}
\def\Dk{\mathsf{D}}
\newcommand{\Hk}{\mathsf H}
\newcommand{\dualHk}{\mathsf H^\perp}

\def\Hmax{H_{\max}}
\def\Hmin{H_{\min}}
\def\ssim{{\scriptscriptstyle \sim}}

\newtheorem{theorem}{Theorem}
\newtheorem{corollary}{Corollary}
\newtheorem{prop}{Proposition}
\newtheorem{lemma}{Lemma}

 \def\topbotatom#1{\hbox{\hbox to 0pt{$#1\bot$\hss}$#1\top$}}
   \newcommand*{\topbot}{{\mathord{\mathchoice{\topbotatom\displaystyle}
                                    {\topbotatom\textstyle}
                                    {\topbotatom\scriptstyle}
                                    {\topbotatom\scriptscriptstyle}}}}

\usepackage{pgfplots}
 \pgfplotsset{compat=newest}
 \usetikzlibrary{intersections}
 \usetikzlibrary{external}
\usepackage{xfrac}

\usepackage{setspace}
\usepackage{titling}

\posttitle{\par\end{center}}
\predate{}
\postdate{}
\begin{document}

\title{\large {\bf Duality of channels and codes}}

\author{
{\normalsize 
Joseph M.\ Renes}\\
\emph{\normalsize 
Institute for Theoretical Physics, ETH Z\"urich, Switzerland}
}

\date{\vspace{-\baselineskip}}

\maketitle

\begin{abstract}
For any given channel $W$ with classical inputs and possibly quantum outputs, a dual classical-input channel $W^\perp$ can be defined by embedding the original into a channel $\cN$ with quantum inputs and outputs.  
Here we give new uncertainty relations for a general class of entropies that lead to very close relationships between the original channel and its dual. 
Moreover, we show that channel duality can be combined with duality of linear codes, whereupon the uncertainty relations imply that the performance of a given code over a given channel is entirely characterized by the performance of the dual code on the dual channel. 
This has several applications. 
In the context of polar codes, it implies that the rates of polarization to ideal and useless channels must be identical. 
Duality also relates the tasks of channel coding and privacy amplification, implying that the finite blocklength performance of extractors and codes is precisely linked, and that optimal rate extractors can be transformed into capacity-achieving codes, and vice versa. 
Finally, duality also extends to the EXIT function of any channel and code. 
Here it implies that for any channel family, if the EXIT function for a fixed code has a sharp transition, then it must be such that the rate of the code equals the capacity at the transition. 
This gives a different route to proving a code family achieves capacity by establishing sharp EXIT function transitions.
\end{abstract}
\vspace{0.5\baselineskip}

\section{Introduction}
Duality is an important concept in many branches of mathematics, often enabling given problems to be transformed into dual versions that are simpler to solve. 
Recently, the author and collaborators have introduced a dual channel in the context of quantum information processing and polar coding~\cite{renes_physical_2008,renes_duality_2011,renes_polar_2014,renes_alignment_2016}. 
The dual construction applies to channels with classical inputs and classical or quantum outputs and is designed so that the original channel and its dual can both be embedded into the same quantum channel. 
Constraints on the form of quantum channels then lead to nontrivial constraints on the behavior of the channel and its dual. 

Here we investigate the notion of duality more comprehensively. 
We find that it is entirely compatible with the duality of linear codes generally, as well as the notion of channel convolution appearing in belief propagation decoding and polar coding more specifically. 
Entropic uncertainty relations imply constraints between a wide variety of entropic functions of the channel and code, including EXIT functions. 
As the class of entropies is quite large, including R\'enyi entropies for instance, this essentially means that the behavior of a code over a channel is determined by that of the dual code over the dual channel. 

Channel duality has several applications, which we briefly describe here by way of outlining the structure of the paper. 
In the next section we set the mathematical stage and define the class of entropies under consideration. 
Section~\ref{sec:duality} is then concerned with the definition and properties of dual channels themselves. 
In particular, duals of simple classical channels are given, and its relation with channel convolution is established in Theorem~\ref{thm:convolution}. 
The main result of \S\ref{sec:duality}, Theorem~\ref{thm:eurchannel} is a tight entropic uncertainty relation between a channel and its dual. 
Some implications of this relation are given, such as the precise tradeoff of channel capacities and equality of channel dispersions. 
Perhaps more importantly, Theorem~\ref{thm:eurchannel} also implies that the rates of polarization of arbitrary channels to either the ideal or useless channel are in fact identical; this is stated precisely in Corollary~\ref{cor:polarizationrate}.  

Section~\ref{sec:codechannel} considers duality for codes and channels. 
After examining the notion of code duality in the quantum-mechanical setting, the dual of a encoder and channel combination is shown to be related to randomized encoding of the dual channel in Proposition~\ref{prop:dualencoder}. 
The tight entropic relation for channels is then extended to the channel and code case in Theorem~\ref{thm:eurchannelcode}. 
This implies that channel coding and privacy amplification are closely related, so that randomness extractors can be used to create channel codes, and vice versa, where the error probability of the code is precisely related to the quality of the extracted key. 
Morever, the optimal finite-blocklength sizes of channel codes and randomness extractors sum precisely to the blocklength, as seen in Corollary~\ref{cor:finiteblock}. 
This can be used to sharpen bounds on finite blocklength bounds on randomness extraction as illustrated in an example. 
Finally, duality for EXIT functions is shown in Theorem~\ref{thm:exit}. 
Combined with capacity duality, it implies that sharp transitions in the EXIT function must occur ``at capacity'', i.e.\ at a noise parameter such that the corresponding capacity equals the rate of the chosen code. This replaces the area theorem in locating the transition, as used for instance in the proof by Kudekar \emph{et al.}\ that Reed-Muller codes achieve capacity over erasure channels~\cite{kudekar_reed-muller_2016-1}. 


\section{Preliminaries}
\label{sec:mathback}
\subsection{Mathematical setup}
First let us fix the notation used to describe classical random variables, quantum states, and channels of all kinds. 
For a random variable $X$ over alphabet $\cX$, we denote its probability distribution by $P_X$ and the size of its alphabet by $|\cX|$. 
The Hilbert space associated with a quantum system $A$ is denoted $\cH_A$ and its dimension $|A|$. 
The set of density operators on $\cH_A$, i.e.\ the positive semidefinite linear maps from $\cH_A$ to itself having unit trace, is denoted $\cD(\cH_A)$.
A channel $W$ which takes $z\in \cZ$ to $W(z)\in\cD(\cH)$ is called a classical-quantum or CQ channel. 
A fully quantum channel, say from $\cD(\cH_A)$ to $\cD(\cH_B)$, will be denoted $\cE_{B|A}$. This notation mimics the notation for conditional probability distributions, which are channels from classical systems (random variables) to classical systems. 
In the same spirit, just as $P_X$ is the marginal probability of $X$ when working in the context of the joint distribution $P_{XY}$, for quantum states $\rho_A$ is the marginal state when working in the context of a joint state $\rho_{AB}\in \cD(\cH_A\otimes \cH_B)$, i.e. $\rho_A=\tr_B[\rho_{AB}]$. 
We also occasionally abuse notation by referring to a pure state (density operator of rank one) by its nonzero eigenvector in situations calling for density operators, e.g.\ $\ket{\sigma}_A$ as the input to a channel. 

The construction of the dual makes use of two conjugate bases of the input Hilbert space, which are defined using the discrete Fourier transform. 
Let $\{\ket z\}_{z=0}^{d-1}$ be an arbitrary basis of some $\cH$ of dimension $d$ and then define the conjugate basis with elements $\cket x=\tfrac1{\sqrt{d}}\sum_{z=0}^{d-1}\omega^{xz}\ket{z}$ for $x=0,\dots,d-1$ and $\omega$ a primitive $d$th root of unity.
We will also make use of the operators $X=\sum_{z=0}^{d-1}\ket{z+1}\bra z$ and $Z=\sum_{z=0}^{d-1}\omega^z\ketbra z$, and we will refer to the $\ket z$ basis as the standard basis and $\cket{x}$ as the conjugate basis. 
Arithmetic inside kets is understood to be modulo $d$. 
Observe that $X=\sum_{x=0}^{d-1}\omega^{-x}\cket{x}\cbra{x}$. 
The canonical maximally entangled state on $\cH_{A}\otimes \cH_{A'}$ with $\cH_{A'}\simeq \cH_A$ associated with the standard basis is $\ket{\Phi}_{AA'}=\tfrac1{\sqrt{|A|}}\sum_z \ket z_A\ket z_{A'}=\tfrac{1}{\sqrt{|A|}}\sum_x \cket{x}_A\ket{{-}\tilde x}_{A'}$.  

The use of the Fourier transform is related to treating the input alphabet $\cZ$ of CQ channels as an Abelian group. 
In this setting, it is natural to consider a CQ channel $W$ to be symmetric if there exists a set of unitary transformations $U_z$ for $z\in \cZ$ such that $U_{z'}W(z)U_{z'}^*=W(z+z')$ for all $z,z'\in \cZ$.
Here $U^*$ denotes the adjoint of the map $U$.
Two CQ channels $W:\cZ\to \cD(\cH)$ and $W':\cZ\to\cD(\cH') $ are said to be output equivalent if there exist quantum channels $\cE:\cD(\cH)\to \cD(\cH')$ and $\cF:\cD(\cH')\to \cD(\cH)$ such that $\cE(W(z))=W'(z)$ and $\cF(W'(z))=W(z)$ for all $z\in \cZ$. 
We will regard them as equivalent, denoted $W\simeq W'$, if there exists a bijection $T$ on $\cZ$ such that $W'\circ T$ is output equivalent to $W$. 
Equivalence in this sense is that of equivalence for coding purposes. 

The properties of quantum channels will also be important in the construction of the dual channel. 
The most important of these is the Stinespring representation theorem, which states that any channel $\cE_{B|A}:\cD(\cH_A)\to \cD(\cH_B)$ can be represented by an isometry from $\cH_A$ to the joint system $\cH_{BE}=\cH_B\otimes \cH_E$, followed by discarding $E$.  Formally, $\cE_{B|A}(\rho_A)=\tr_E[V_{BE|A}\rho_A V_{BE|A}^*]$ for some isometry $V_{BE|A}$ and all $\rho_A\in \cD(\cH_A)$. 
Using the representative isometry we can define the complementary channel $\cE^\sharp_{E|A}$ by the action $\cE^\sharp_{E|A}(\rho_A)=\tr_B[V_{BE|A}\rho_A V_{BE|A}^*]$ for all $\rho_A\in \cD(\cH_A)$. (Note that $\sharp$ does not operate on the system labels of the channel, but $*$ does.)
The representative isometry is not unique, but for any two isometries $V_{BE|A}$ and $V'_{BE'|A}$ associated with the same channel there exists a partial isometry $U_{E'|E}$ such that $V'_{BE'|A}=U_{E'|E}V_{BE|A}$. (We only require $U_{E'|E}^*U_{E'|E}$ to be a projection onto the image of $V_{BE|A}$.)
Thus, the complementary channels are also not unique, though essentially so, as they are all related by the action of partial isometries on the output system.

\subsection{Dual entropies}
Entropy duality will also play a crucial role in the results, which hold for a wide variety of entropy measures. 
Following \cite{coles_uncertainty_2012}, let $\Dk(\rho,\sigma)$ for $\rho\in\cD(\cH)$ and $\sigma$ a positive operator on $\cH$ be a divergence measure which satisfies the following four properties: 
\begin{enumerate}
\item Monotonicity, or the data-processing inequality: For any channel $\cE$, $\Dk(\rho,\sigma)\geq D(\cE(\rho),\cE(\sigma))$,
\item Normalization: For $c>0$, $\Dk(\rho,c\sigma)=\Dk(\rho,\sigma)-\log c$,
\item Dominance: For $\sigma'\geq \sigma$, $\Dk(\rho,\sigma')\leq\Dk(\rho,\sigma)$, and 
\item Zero: $\Dk(\rho,\rho)=0$. 
\end{enumerate}
Using any $\Dk$ we may define two conditional entropies 
\begin{align}
\label{eq:entropydown}
\Hk_\downarrow (A|B)_\rho&:=-\Dk(\rho_{AB},\id_A\otimes \rho_B)\,, \qquad \text{and}\\
\Hk_\uparrow(A|B)_\rho&:=\max_\sigma[-\Dk(\rho_{AB},\id_A\otimes \sigma_B)]\,.\label{eq:entropyup}
\end{align}
Each has a dual, defined by $\dualHk_\downarrow(A|B)_\rho:=-\Hk_\downarrow(A|C)_\rho$ for pure $\rho_{ABC}$, and similarly for $\dualHk_\uparrow(A|B)_\rho$. 

The standard von Neumann entropy $H(A|B)_\rho$, defined using the relative entropy $D(\rho,\sigma)=\tr[\rho(\log \rho-\log \sigma)]$, is self-dual.
The optimal $\sigma_B$ is the marginal $\rho_B$, so $H_\downarrow(A|B)_\rho=H_\uparrow(A|B)_\rho$. 
We will also be interested in the dispersion $V(\rho,\sigma)=\tr[\rho(\log\rho-\log \sigma)^2]$. 

There are two especially useful versions of the R\'enyi entropy in the quantum setting, defined using either the Petz~\cite{petz_quasi-entropies_1986} or sandwiched~\cite{muller-lennert_quantum_2013,wilde_strong_2014} R\'enyi divergences, respectively:
\begin{align}
\bar D_\alpha(\rho,\sigma):=&\tfrac1{\alpha-1}\log \tr[\rho^\alpha\sigma^{1-\alpha}]\,\quad\text{and}\\
\tilde D_\alpha(\rho,\sigma):=& \tfrac1{\alpha-1}\log \tr[(\sigma^{\frac{1-\alpha}{2\alpha}}\rho\sigma^{\frac{1-\alpha}{2\alpha}})^\alpha]\,.
\end{align}
Both of these satisfy the four properties above (for an excellent overview, see~\cite{tomamichel_quantum_2016}).
Various duality relations are known for the various R\'enyi entropies~\cite{tomamichel_fully_2009,konig_operational_2009,muller-lennert_quantum_2013,beigi_sandwiched_2013,tomamichel_relating_2014}. 
In particular, 
\begin{align}
(\bar H_\alpha^\downarrow)^\perp &=\bar H_{2-\alpha}^\downarrow\qquad\qquad \alpha\in[0,2]\\
(\tilde H_\alpha^\uparrow)^\perp &=\tilde H_{\alpha/(2\alpha-1)}^\uparrow\qquad \alpha\in[\tfrac12,\infty]\\
(\bar H_\alpha^\uparrow)^\perp &=\tilde H_{1/\alpha}^\downarrow \qquad\qquad \alpha\in[0,\infty].
\end{align}
Especially useful are the min- and max-entropies, $\Hmin=\tilde H_\infty^\uparrow$ and $\Hmax=\tilde H_{1/2}^\uparrow$, which are dual to one another.
These can also be directly defined by 
\begin{align}
H_{\min}(A|B)_\psi&=\max_{\sigma\in \cD(\cH_B)} \sup \{\lambda\in \mathbb R:\psi_{AB}\leq 2^{-\lambda}\id_A\otimes \sigma_B\}\label{eq:minentropy} \\
{H_{\max}(A|B)_\psi}&= \max_{\sigma\in \cD(\cH_B)}\log |A|\,F(\psi_{AB},\mix_A\otimes \sigma_B)^2\,,\label{eq:maxentropy}
\end{align}
where $F(\rho,\sigma)=\|\sqrt{\rho}\sqrt{\sigma}\|_1$ is the fidelity, the quantum analog of the Bhattacharyya parameter.

In all these examples, the dual entropy is itself known to be a divergence-based entropy for an appropriate choice of divergence. 
But we could also choose different variants of the R\'enyi divergence, for which the duals of the associated conditional entropies are not known to themselves come from a divergence, such as the maximal~\cite{matsumoto_new_2013} or the reversed sandwiched relative entropy~\cite{audenaert_-z-renyi_2015}, respectively\footnote{These can be shown to satisfy dominance using \cite[Lemma 5]{coles_uncertainty_2012}.}
\begin{align}
D^{\text{maximal}}_\alpha&:=\tfrac1{\alpha-1}\log\tr[\sigma^{1/2}(\sigma^{-1/2}\rho\sigma^{-1/2})^\alpha\sigma^{1/2}]\qquad \text{and} \\
D^{\text{reverse}}_\alpha&:=\tfrac1{\alpha-1}\log \tr[\left(\rho^{\frac\alpha{2(\alpha-1)}}\sigma \rho^{\frac\alpha{2(\alpha-1)}}\right)^{1-\alpha}]\,.
\end{align}
Another example is the conditional entropy based on hypothesis testing~\cite{wang_one-shot_2012}.  
Consider the the minimum type-II error in asymmetric hypothesis testing of $\rho$ versus $\sigma$ with fixed type-I error,
\begin{align}
\beta_\eps(\rho,\sigma):=\min\{\tr[\Lambda \sigma]:\tr[\Lambda\rho]\geq 1-\eps,0\leq\Lambda\leq \id\}\,.
\end{align}
Then the divergence $D_h^\eps(\rho,\sigma)=-\log \frac{\beta_\eps(\rho,\sigma)}{1-\eps}$ satisfies the four properties above~\cite{jensen_generalized_2013}.

Beyond the framework of entropy duality based on relative entropies, another dual pair is given by the \emph{smooth} min- and max-entropies. 
To define them, first define the purification distance $P(\rho,\sigma)$ between two states $\rho$ and $\sigma$ to be $P(\rho,\sigma)=\sqrt{1-F(\rho,\sigma)^2}$. 
Then denote by $\cB_\eps(\rho)$ the set of states with distance no larger than $\eps$ from $\rho$. 
Finally, we can define the smooth entropies:
\begin{align}
\Hmin^\eps(A|B)_\rho:=&\max_{\rho'\in \cB_\eps(\rho)} \Hmin(A|B)_{\rho'}\qquad \text{and}\label{eq:smoothmin}\\
\Hmax^\eps(A|B)_\rho:=&\min_{\rho'\in \cB_\eps(\rho)} \Hmax(A|B)_{\rho'}\,.\label{eq:smoothmax}
\end{align}
Min- and max-entropies with identical smoothing parameters are dual to one another~\cite{tomamichel_duality_2010}.

Mostly we will be interested in the entropy of a classical random variable conditional on a quantum system, say $X$ given $B$. 
This is denoted $\Hk(X|B)_\psi$ where the state $\psi_{XB}$ is the CQ state $\psi_{XB}=\sum_x P_X(x)\ketbra x_X\otimes (\rho_x)_B$ corresponding to the ensemble $\{P_X(x),\rho_x\}_x$. 
Often the classical random variable will be the result of measuring a quantum observable, say the observable $Z$ on system $A$. 
Overloading notation somewhat, we denote this random variable $Z_A$ and the  conditional entropy by $\Hk(Z_A|B)_\psi$, where now $\psi$ denotes the state prior to the measurement. 
Formally, if $\rho_{ZB}=\sum_z \ketbra z_Z\otimes \tr_A[\ketbra z_A\psi_{AB}]$, then $\Hk(Z_A|B)_\psi=\Hk(Z|B)_\rho$. 

The min-entropy of a CQ state is directly related to the optimal probability of guessing the value of the random variable by making a measurement of the quantum system~\cite{konig_operational_2009}. 
Formally, for a CQ state $\psi_{XB}$, let 
\begin{align}
\pg(X|B)_\psi:=\max_{\Lambda_x} \sum_x P_X(x)\tr[\Lambda_x \rho_x],
\end{align}
where the optimization is over all POVMs $\{\Lambda_x\}$, i.e.\ sets of positive operators $\Lambda_x$ on $\cH_B$ such that $\sum_{x}\Lambda_x=\id$. 
Then
\begin{align}
\pg(X|B)_\psi=2^{-\Hmin(X|B)_\psi}
\end{align}
Its dual, the max-entropy, is related to the quality of $X$ as a secret key relative to $B$, as the fidelity measures how close the CQ state is to one in which $X$ is uniform and completely independent of $B$. 
This is a form of ``decoupling'' of $X$ from $B$. 
We will denote the decoupling quality as $Q(X|B)_\psi:=\max_{\sigma\in\cD(\cH_B)}F(\psi_{XB},\mu_X\otimes \sigma_B)^2$. 
Then, from \eqref{eq:maxentropy} we have $2^{\Hmax(X|B)_\psi}=|X|Q(X|B)_\psi$. 

\section{Dual channels}
\label{sec:duality}
\subsection{Definition and basic properties}
The notion of a dual channel based on embedding both the original and dual channels into a single quantum channel is implicit in \cite{renes_physical_2008,boileau_optimal_2009,renes_duality_2011}. Here we follow and add detail to the more explicit presentation of \cite{renes_polar_2014,renes_alignment_2016}.
Consider an arbitrary CQ channel $W$ with classical inputs in an alphabet $\cZ$ and quantum outputs which are density operators on the Hilbert space $\cH_B$.
We can embed $W$ in a quantum channel $\cN_{B|A}$ from $A$ to $B$ which measures the quantum input $A$ in the $\ket{z}$ basis and then produces the corresponding output $\varphi_z$ in $B$. 
Formally, this is described by $\cN_{B|A}(\rho_A)=\sum_z \bra z\rho\ket z\, (\varphi_z)_B$. 
The dual channel comes from using the complement of $\cN_{B|A}$, restricted to inputs diagonal in the conjugate basis. 
Formally, 
\begin{align}
\label{eq:dualdefinition}
W^\perp(x):=\cN_{E|A}^\sharp(\ketbra{\tilde x}_A)\,.
\end{align} 
As noted above, the complement is not unique, so the definition in \eqref{eq:dualdefinition} leads to a family of dual channels. 
Nevertheless, since complementary channels are all related by partial isometries, all possible dual channels are equivalent to one another. 
For a convenient concrete representation, let $\ket{\varphi_{z}}_{BD}$ be a purification of $\varphi_z$ and define the isometry $V_{BCD|A}$ by  
\begin{equation}
\label{eq:channelisometry}
V_{BCD|A}\ket z_A= \ket z_C\otimes \ket{\varphi_z}_{BD}\,.
\end{equation}
Here $CD$ together form the dilation space $E$. 
Defining $\ket{\theta_x}_{BCD}:=V_{BCD|A}\cket x_A$, the channel outputs are simply $W^\perp(x)=(\theta_x)_{CD}$. 

It is also useful to note that we can generate the outputs of $W$ and $W^\perp$ from the following maximally-entangled quantum state $\ket{\psi}_{ABCD}$ by measuring system $A$ appropriately. 
Using the two expressions for $\ket{\Phi}_{AA'}$ in the standard and conjugate bases, we have 
\begin{subequations}
\label{eq:channelstate}
\begin{align}
\ket{\psi}_{ABCD}&=\id_A\otimes V_{BCD|A'}\ket{\Phi}_{AA'}\\
&=\tfrac1{\sqrt{d}}\sum_z \ket{z}_A\ket{z}_{C}\ket{\varphi_z}_{BD}\\
&=\tfrac1{\sqrt{d}}\sum_x \ket{{-}\tilde x}_A\ket{\theta_x}_{BCD}\,.
\end{align} 
\end{subequations}
where we now take $V$ to act on $\cH_{A'}$. 
By the definition of $\ket{\Phi}_{AA'}$, measurement of $A$ in the standard basis $\{\ket z\}$ clearly yields $\ket{z}_C\ket{\varphi_z}_{BD}$ for outcome $z$, and subsequently tracing out $CD$ gives $W(z)$. 
Meanwhile, measurement of $A$ in the conjugate basis $\{\cket x\}$ yields $\ket{\theta_x}_{BCD}$ for outcome $-x$, and subsequently tracing out $B$ gives $W^\perp(x)$. 
Indeed, we could just as well take this procedure of starting from \eqref{eq:channelstate} and measuring appropriately as the definition of the dual, and we will frequently make use of this formulation in the remainder of the paper. 

Equivalent channels $W$ and $W'$ have equivalent duals:
\begin{prop}
\label{prop:equiv}
For any two CQ channels $W$ and $W'$ such that $W\simeq W'$, it holds that $W^\perp\simeq W'^\perp$. 
\end{prop}
\begin{proof}
Let $\cN_{B|A}$ and $\cN'_{B'|A}$ be the quantum channels associated to $W$ and $W'$,  respectively.
By equivalence, there exist channels $\cE_{B'|B}$ and $\cE'_{B|B'}$ such that $\cE_{B'|B}\circ\cN_{B|A}=\cN'_{B'|A}$ and $\cE'_{B|B'}\circ\cN'_{B'|A}=\cN_{B|A}$. 
Now suppose $V_{BE|A}$ and $V'_{B'E'|A}$ are Stinespring dilations of $\cN_{B|A}$ and $\cN'_{B'|A}$, while $U_{B'D'|B}$ and $U'_{BD|B'}$ are dilations of $\cE_{B'|B}$ and $\cE'_{B|B'}$. 
By the first equivance statement, there must exist a partial isometry $T'_{D'E|E'}$ such that $U_{B'D'|B}V_{BE|A}=T'_{D'E|E'} V'_{B'E'|A}$. 
Similarly, by the second, there exists a partial isometry $T_{DE'|E}$ such that $U'_{BD|B'}V'_{B'E'|A}=T_{DE'|E} V_{BE|A}$. 
Hence we can define $\cF_{E'|E}(\cdot)=\tr_{D}[T(\cdot)T^*]$ and $\cF'_{E|E'}(\cdot)=\tr_{D'}[T'(\cdot)T'^*]$ to satisfy $W'^\perp=\cF\circ W^\perp$ and $W^\perp=\cF'\circ W'^\perp$, where $V$ and $V'$ are used to define the complements in the duals $W^\perp$ and $W'^\perp$. 
\end{proof}

Because the channel $\cN_{B|A}$ measures the input in the $\ket z$ basis, the outcome $\ket{z}$ shows up in the Stinespring isometry. 
Therefore it is in some sense copied to the output of $W^\perp$. 
This leads to symmetry of the dual channel, which is present even if the original channel $W$ is not symmetric. 
Specifically, an easy calculation shows that $\ket{\theta_x}_{BCD}=Z_C^x\ket{\theta_0}_{BCD}$, and therefore the value of $x$ modulates system $C$ with the unitary operator $Z$ and doesn't involve $B$ or $D$.
Thus, the dual channel has a simple group covariance structure $W^\perp(x)=Z^x_C W^\perp(0) Z^{-x}_C$, irrespective of the properties of $W$.

This also immediately implies that $(W^\perp)^\perp\nsimeq W$ in general.
However the dual of the dual is the symmetrized version $W_{\text{sym}}$ of $W$, in the sense of \cite[Definition 1.3]{korada_polar_2009}. 
More specifically, for a general CQ channel $W$, let $W_{\text{sym}}$ be defined by $W_{\text{sym}}(z)=\tfrac1{|\cZ|}\sum_{z'}\ketbra{z+z'}\otimes W(z')$.  
Then we have
\begin{prop}
For any CQ channel $W$, $(W^\perp)^\perp\simeq W_{\text{sym}}$. 
If $W$ is symmetric, then $(W^\perp)^\perp\simeq W$. 
\end{prop}
\begin{proof}
Iterating the above construction of the dual, it follows that the output of $(W^\perp)^\perp$ for input $y$ is the $B'B$ marginal of the state 
\begin{subequations}
\begin{align}
\ket{\xi_y}_{BB'CD}
&=\tfrac{1}{\sqrt{d}}\sum_x \omega^{xy}\ket{x}_{B'}\ket{\theta_x}_{BCD}\\
&=\tfrac1d\sum_{xz}\omega^{x(y+z)}\ket{x}_{B'}\ket{z}_C\ket{\varphi_z}_{BD}\,.
\end{align} 
\end{subequations}
Direct calculation gives $(\xi_y)_{BB'}=\tfrac1d\sum_z \ketbra{\widetilde {y{+}z}}_{B'}\otimes \varphi_z$, which is the output of $W_{\text{sym}}(y)$, up to Fourier transform on $B'$.

For symmetric $W$, suppose we perform a controlled-unitary operation on $B'B$ which applies $U_{y{+}z}^*$ to $B$ when $B'$ is in the state $\ketbra{\widetilde {y{+}z}}$. 
This produces $(\xi'_y)_{B'B}=\mu_{B'}\otimes (\phi_{-y})_B$. 
Since we can invert the input to $W$ and append $\mu_{B'}$ to obtain this state, $(W^\perp)^\perp$ is equivalent to $W$. 
\end{proof}

\subsection{Duals of classical channels}
The dual of a classical channel has a particular form. 
Suppose that the $\varphi_z$ are determined by a conditional probability distribution $P_{Y|Z}$ in the sense that $\varphi_z=\sum_y P_{Y|Z=z}(y)\ketbra y$, and define the unnormalized states $\ket{\eta_y}=\tfrac 1{\sqrt d}\sum_z \sqrt{P_{Y|Z=z}(y)}\ket{z}$.
Computing $V\ket{\tilde 0}$, we find 
\begin{align}
V\ket{\tilde 0}
&=\sum_{y}\ket{\eta_y}_C\ket{y}_B\ket{y}_D.
\end{align}
Observe that the norm $\bracket{\eta_y}{\eta_y}$ is just $\tfrac 1d\sum_z P_{Y|Z=z}(y)$, i.e.\ $P_Y(y)$ assuming that $Z$ is uniformly distributed.
Now we can write a useful form for $\theta_x$:
\begin{align}
\label{eq:ccdualoutput}
\theta_x=\sum_y Z^x_C\ketbra{\eta_y}_CZ_C^{-x}\otimes \ketbra y_D.
\end{align}
Thus, the dual of a classical channel outputs two systems, one classical and one quantum. 
The former (system $D$) records the classical value of $y$, while the latter (system $C$) is a pure state $\ket{\eta_y}$ which has been modulated by $Z$ according to the value of $x$. 

In case the quantum outputs of the channel are commuting states, we can regard the channel as a classical channel. 
This can only happen in two ways: either $P_{Z|Y=y}$ is uniform, in which case the outputs are orthogonal states, or $P_{Z|Y=y}$ is concentrated on only one value of $Z$, and the output is the same for all inputs. 
In the former case, the input is completely recoverable from the output in $C$, while for the latter recovery better than blind guessing is completely impossible. 
Thus the only classical channels which have classical duals are erasure-like channels in which the output $y$ in $D$ indicates whether the input is perfectly recoverable from $C$ or has been essentially erased.  

For binary-input channels, since the quantum outputs $Z^x\ket{\eta_y}$ are pure states, they are completely characterized by their overlap
$\cos\vartheta_y=|\bra{\eta_y}Z\ket{\eta_y}|/\bracket{\eta_y}{\eta_y}$. 
Working this out explicitly, one finds  
\begin{align}
\label{eq:overlap}
\cos\vartheta_y=|P_{Z|Y=y}(0)-P_{Z|Y=y}(1)|\,,
\end{align} 
which is just the \emph{$|D|$-density} of the original binary-input classical channel~\cite[Eq.\ 4.12]{richardson_modern_2008}.
The duals of binary symmetric and erasure channels are computed in~\cite{renes_alignment_2016}, and the results can be immediately understood using \eqref{eq:overlap}. 
The BEC is its own dual, in the sense that $\text{BEC}(p)^\perp=\text{BEC}(1{-}p)$. 
This can be seen because the overlaps are either zero ($y=?$) or one ($y=0,1$), corresponding to quantum outputs $\tfrac1{\sqrt2}(\ket{0}+(-1)^x\ket 1)$  or $\ket{0}$ respectively. 
Thus, the value of $x$ is perfectly recoverable when $y=?$ but not at all when $y=0,1$. 
The dual of the BSC, meanwhile, has pure state outputs (up to equivalence), taking $x$ to $\sqrt{p}\ket{0}+(-1)^x\sqrt{1-p}\ket{1}$. 
Again using \eqref{eq:overlap}, it is clear that the overlap $|1-2p|$ is the same for both values of $y$, so we can dispense with this part of the output of the dual. 
Comparing to the form of \eqref{eq:ccdualoutput}, the fact that that any symmetric binary-input classical channel can be thought of as a heralded mixture of BSCs (see, e.g.~\cite{land_bounds_2005}) is reflected in the fact that its dual is a heralded mixture of BSCs. 

Constructing the dual of the binary-input additive white Gaussian noise channel, $W:z\to y=z+N$ with $N$ normally-distributed, is more subtle, because strictly speaking the framework above does not apply.
Since the outputs are continuous, we would like to use $\ket{y}$ for $y\in \mathbb{R}$, but this is not a proper basis set. 
Put differently, $\phi_z=\int{\text{d}}y P_{Y|Z=z}(y)\ketbra y$ is not a proper density operator. 
Nonetheless, the ultimate result will look essentially the same: The dual $W^\perp$ will take $x$ to a joint classical-quantum system, a classical random variable $Y$ governed by $P_{Y}(y)$ and an associated qubit $C$ in the state $Z^x\ket{\eta_y}$. The precise details of the construction will be reported elsewhere. 

\subsection{Extremality of the BSC and its dual}

The BSC and its dual are extremal binary-input channels in the following sense. 
First, any such channel can be degraded to a BSC, simply by performing the optimal measurement for distinguishing the outputs. 
This operation preserves the trace distance (the quantum analog of the variational distance) of the two outputs, $\delta(W):=\tfrac12\|\phi_0-\phi_1\|_1$, since the optimal measurement is known to have an error probability of $\tfrac12(1-\delta(W))$~\cite{helstrom_detection_1967,helstrom_quantum_1976}.
Let us denote this channel by $W_\text{BSC}$. 
Similarly, any binary-input channel can be upgraded to a channel with pure state outputs, simply by finding pure states $\ket{\Psi_j}_{BB'}$ such that $\phi_j=\tr_B[(\Psi_{j})_{BB'}]$ for $j=0,1$. 
By Uhlmann's theorem~\cite{uhlmann_transition_1976} (see also \cite[Theorem 9.4]{nielsen_quantum_2000}), it is possible to find purifications such that $F(\phi_0,\phi_1)=|\bracket{\Psi_0}{\Psi_1}|$, so this construction preserves the fidelity.
Let us denote the resulting pure state channel by $W_\text{pure}$, and the fidelity of the outputs of any binary-input channel symmetric $W$ as $F(W)$.  

Duality relates these two constructions in an elegant way. 
To see this, we first state a result shown in the proof of \cite[Proposition 3.6]{renes_alignment_2016}, and we include the proof here for completeness.
\begin{prop}
\label{prop:deltaF}
For any binary-input symmetric channel $W$, $\delta(W)=F(W^\perp)$. 
\end{prop}
\begin{proof}
Using Uhlmann's theorem and \eqref{eq:channelstate}, for unitary $U$ we have,
\begin{subequations}
\begin{align}
F(W^\perp)
&=\max_U |\bra{\theta_0}U_B\ket{\theta_1}_{BCD}|\\
&=\max_U \tfrac12 \Big|\sum_{z\in \{0,1\}} (-1)^z \bra{\varphi_z}U_B\ket{\varphi_z}_{BD}\Big|\\
&=\max_U\tfrac12|\tr[U(\phi_0-\phi_1)]|
\end{align}
\end{subequations}
Since $\|A\|_1=\max_U |\tr[UA]|$, the desired result follows. 
\end{proof}
Now degrade $W^\perp$ to a BSC. 
We have $\delta(W^\perp)=\delta((W^\perp)_\text{BSC})$ by the properties of the degrading map, which together with Proposition~\ref{prop:deltaF} implies $F(W)=F(((W^\perp)_\text{BSC})^\perp)$. 
Since $((W^\perp)_\text{BSC})^\perp$ is a pure state channel with the same fidelity as $W$, it is necessarily equivalent to $W_\text{pure}$. 
Therefore, we have shown
\begin{corollary}
For any binary-input symmetric CQ channel $W$, $W_\text{pure}\simeq ((W^\perp)_\text{BSC})^\perp$.
\end{corollary}

\subsection{Channel convolution}
As discussed in \cite{renes_alignment_2016}, the dual is compatible with the notion of channel convolution appearing in the setting of polar codes. 
Essentially the same notion also appears in belief propagation decoding of general binary linear codes, as the update rules for messages at check and variable nodes~\cite{richardson_modern_2008}.  
The check $(\boxasterisk)$ and variable $(\oasterisk)$ convolutions are defined by 
\begin{align}
[W \oasterisk W'](z)&:=W(z)\otimes W'(z),\\
[W\boxasterisk W'](z)&:=\tfrac12W(z)\otimes W'(0)+\tfrac12 W(z+ 1)\otimes W'(1).
\end{align}
In the context of polar coding, the check convolution is precisely the ``worse'' channel synthesized from $W$ and $W'$, call it $W\boxminus W'$, while for symmetric $W$ and $W'$, the variable convolution is equivalent to the ``better'' synthesized channel $W\boxplus W'$. 
Formally,
\begin{align}
W\boxminus W'&= W\boxasterisk W'\qquad \text{and}\\
W\boxplus W'&\simeq W\oasterisk W'\,.
\end{align}
To see the latter, first observe that the better channel has outputs  
\begin{align}
[W\boxplus W')](z)=\tfrac12\ketbra 0\otimes W(z)\otimes W'(z)+\tfrac12\ketbra 1\otimes W(z+ 1)\otimes W'(z)\,.
\end{align}
Symmetry of $W$ amounts to the existence of a unitary operator $U$ such that $W(z+ 1)=U W(z) U^*$ for $z=0,1$. 
Therefore, applying $U$ if the first system is in the state $\ket{1}$ and doing nothing otherwise results in the state $\tfrac12\sum_{z'}\ketbra {z'}\otimes W(z)\otimes W'(z)$ for input $z$. 
Since the first system is independent of the second two, we have $W\boxplus W'\simeq W\oasterisk W'$ for symmetric $W$, as intended. 

The compatibility of convolution with the dual is the following theorem.
\begin{theorem}
\label{thm:convolution}
For any two binary-input CQ channels $W$ and $W'$,
\begin{align}
(W \oasterisk W')^\perp&\simeq W^\perp \boxasterisk W'^\perp\qquad \text{and}\\
(W \boxasterisk W')^\perp&\simeq W^\perp \oasterisk W'^\perp\,.
\end{align}
\end{theorem}
\begin{proof}
Let $\varphi_z$ and $\varphi_z'$ be the outputs of $W$ and $W'$, and similarly $\theta_x$ and $\theta_x'$ the outputs of $W^\perp$ and $W'^\perp$, respectively. 
Now consider the states $\ket{\psi}$ and $\ket{\psi'}$ from \eqref{eq:channelstate} associated with $W$ and $W'$, and denote the respective systems involved by $ABCD$ and $A'B'C'D'$. 
Applying a \textsc{cnot} operation from $A'$ to $A$ yields
\begin{align}
\ket{\eta} &= U^{\textsc{cnot}}_{A'\to A}\ket{\psi}_{ABCD}\ket{\psi'}_{A'B'C'D'}\\
&=\tfrac 12\sum_{z,z'}\ket{z+ z'}_{A}\ket{z'}_{A'}\ket{z}_{C}\ket{z'}_{C'}\ket{\varphi_{z}}_{BD}\ket{\varphi_{z'}'}_{B'D'}\,.\label{eq:convchannelstate}
\end{align} 
In the conjugate basis the \textsc{cnot} gate has the same action as in the standard basis, but with control and target reversed. 
Therefore we may write
\begin{align}
\ket{\eta}=\tfrac12\sum_{x,x'}\ket{x}_{A}\ket{{x+ x'}}_{A'}\ket{\theta_{x}}_{BCD}\ket{\theta_{x'}'}_{B'C'D'}\,,
\label{eq:convchannelstate2}
\end{align}
where we have abused notation by omitting tildes on the $A$ and $A'$ basis states to denote use of the conjugate basis. 

The outputs of $W\oasterisk W'$ can be generated from this state by measuring system $A'$ in the $\ket{z}$ basis, discarding the $C$ and $D$ systems, and making use of channel symmetry. 
In the binary-input setting, symmetry amounts to the existence of a unitary operator $U$ such that $W(z+ 1)=U W(z) U^*$ for $z=0,1$. 
Applying $U$ to $B$ conditional on the value of $z+z'$ in $A$ therefore gives the state
\begin{align}
\ket{\eta'}=
\tfrac 12\sum_{z,z'}\ket{z+ z'}_{A}\ket{z'}_{A'}\ket{z}_{C}\ket{z'}_{C'}\ket{\varphi_{z'}}_{BD}\ket{\varphi'_{z'}}_{B'D'}\,.
\end{align}
Measuring $A'$ and discarding $CD$ gives output states $\tfrac12 \sum_{z}\ketbra {z}_{A}\otimes (\varphi_{z'})_{B}\otimes (\varphi_{z'}')_{B'}$ for measurement result $z'$. Since the $A$ part of the state is independent of the rest, the outputs are equivalent to $[W\oasterisk W'](z')$. 
By Proposition~\ref{prop:equiv}, the outputs of the dual can therefore be obtained by measuring system $A'$ of $\ket{\eta'}$ in the Fourier-conjugate basis and discarding systems $ABB'$. But since $\ket{\eta'}$ differs from $\ket{\eta}$ only by a unitary action on $AB$, which will anyway be discarded, the outputs of the dual can just as well be obtained from $\ket{\eta}$. 
Using \eqref{eq:convchannelstate2}, these are easily seen to be just $[W^\perp\boxasterisk W'^\perp](x)$.  

For the second statement, return to \eqref{eq:convchannelstate} and note that measuring $A$ in the $\ket z$ basis and discarding $A'CC'DD'$ gives the states $[W\boxasterisk W'](z)$. 
Thus, the outputs of $(W\boxasterisk W')^\perp$ can be generated by measuring $A$ in the $\cket{x}$ basis and discarding $BB'$, for which it is convenient to use \eqref{eq:convchannelstate2}. Again using channel symmetry to shift the index $x'$ to $x$ in $\ket{\theta'_{x'}}_{B'C'D'}$, the outputs are easily seen to be equivalent to those of $(W^\perp\oasterisk W'^\perp)$.  
\end{proof}

In the context of polar coding over memoryless channels, one considers repeated convolution of a channel with itself, with a random choice of which convolution to use at each step. 
Theorem~\ref{thm:convolution} immediately gives a duality relation, a weaker version of which was recently used by the author and collaborators to study the capability of polar codes constructed for a given channel to be used for another~\cite{renes_alignment_2016}. 
Suppose $y^n\in \{0,1\}^n$ for integer $n>0$ and let $\bar y^n=1^n+ y^n$ (understood modulo 2), where $1^n$ is the length-$n$ string of 1s. 
Then define $W_{y^n}$ recursively as $W_{y^{n-1}}\oasterisk W_{y^{n-1}}$ if $b_n=0$ and 
$W_{y^{n-1}}\boxasterisk W_{y^{n-1}}$ if $b_n=1$. 
Repeatedly applying Theorem~\ref{thm:convolution} gives the following:
\begin{corollary}
\label{cor:convolutionduality}
For any symmetric CQ channel $W$, $(W_{y^n})^\perp\simeq (W^\perp)_{\bar y^n}$.
\end{corollary}
\noindent This is an improvement over \cite{renes_alignment_2016}, which showed that $(W^\perp)_{\bar y^n}$ is a degraded version of $(W_{y^n})^\perp$.

\subsection{Entropic relations between a channel and its dual}
\label{sec:entropychannel}
Entropic uncertainty relations constrain the behavior of a channel by that of its dual. In fact, due to the use of conjugate bases and the form the dual, the entropic uncertainty relations hold with equality, not just as inequalities, as is generally the case. Thus, the behavior of a channel is in fact completely characterized by that of its dual.

For symmetric channels, we are often interested in the conditional entropy of the input given the output, assuming uniform inputs; for the von Neumann or Shannon entropy this leads to the formula for capacity. 
Let us define $\Hk(W):=\Hk(Z|W(Z))$ for any of the entropy functions considered in \S\ref{sec:mathback}. Then we have 
\begin{theorem}
\label{thm:eurchannel}
For any CQ channel $W$ with input $Z$ and any conditional entropy $\Hk$,
\begin{align}
\label{eq:eurchannel}
\Hk(W)+\dualHk(W^\perp)=\log |Z|\,.
\end{align}
\end{theorem}
The proof is based entirely on the following uncertainty equality for the kinds of tripartite states that are found in the state-based definition of the dual channel.  
Indeed, using \eqref{eq:channelstate}, both entropy terms can be computed from the state $\ket{\psi}_{ABCD}$: $\Hk(W)=\Hk(Z_A|B)_\psi$ while $\dualHk(W^\perp)=\Hk(X_A|CD)_\psi$. 
Invoking the following lemma with $E=B$ and $F=CD$ gives \eqref{eq:eurchannel}. Its proof is given in Appendix~\ref{app:eur}.
\begin{lemma} 
\label{lem:uncertaintyequality}
For any tripartite pure state $\ket{\psi}_{AEF}$ in which the unnormalized conditional states $(\sigma_z)_F:=\tr_{AE}[\ketbra z_A\psi_{AEF}]$ are pairwise disjoint, i.e.\ $\sigma_z\sigma_{z'}=0$ for $z\neq z'$, we have
\begin{align}
\Hk(Z_A|E)_\psi+\dualHk(X_A|F)_\psi=\log |A|\,,
\end{align}
for $\Hk$ any entropy defined as in \eqref{eq:entropydown}, \eqref{eq:entropyup}, \eqref{eq:smoothmin}, or \eqref{eq:smoothmax}. 
\end{lemma}
Theorem~\ref{thm:eurchannel} has several important implications. 
First, from duality of min- and max-entropy, the guessing probability of the channel is directly related to the decoupling of the dual channel. 
For $W$ taking $Z$ to $B$, let $\pg(W)=\pg(Z|W(Z))$ and $Q(W)=Q(Z|W(Z))$, with $Z$ uniformly distributed. 
Then, taking $\Hk$ to be $\Hmin$ and $\Hmax$, respectively, leads to 
\begin{corollary}
\label{cor:PQ}
For any CQ channel $W$, 
\begin{align}
P(W)&=Q(W^\perp)\quad \text{and}\label{eq:pfromq}\\
Q(W)&=P(W^\perp)\,.\label{eq:qfromp}
\end{align} 
\end{corollary}

Second, from self-duality of the von Neumann entropy, the capacity of the channel is determined by the capacity of the dual, and vice versa. 
For $I(W):=\max_{P_Z}\left(H(Z)-H(Z|W(Z))\right)$, since the optimal input distribution for symmetric channels is the uniform distribution, we have
\begin{corollary}
For any symmetric CQ channel $W$ with input $Z$, 
\begin{align}
\label{eq:capacityduality}
I(W)+I(W^\perp)=\log |Z|\,.
\end{align}
\end{corollary}

Moreover, the duality of R\'enyi entropies implies that the dispersions of a channel and its dual are identical. 
The channel dispersion determines the second order asymptotic behavior of the maximal achievable communication rate as a function of blocklength for large blocklength~\cite{hayashi_information_2009,polyanskiy_channel_2010-1,tomamichel_second-order_2015}, just as the capacity determines the first order behavior.  
To define the dispersion, let $V(Z|B)_\psi:=V(\psi_{ZB},\id_Z\otimes \psi_B)$ and $V(W):=V(Z|W(Z))$. 
Then we have
\begin{corollary}
For any symmetric CQ channel $W$,
\begin{align}
\label{eq:dispersionequality}
V(W)=V(W^\perp)\,.
\end{align}
\end{corollary}
This follows by using $\Hk=\bar H_\alpha^{\downarrow}$, which leads to 
\begin{align}
\bar D_\alpha(\psi_{Z_AB},\id_{Z_A}\otimes \psi_B)+\bar D_{2-\alpha}(\psi_{X_ACD},\id_{X_A}\otimes \psi_{CD})=\log |A|\,.
\end{align}
Then making use of the following, Proposition 4 in \cite{lin_investigating_2015}, gives the desired result. 
\begin{align}
\label{eq:dispersionderivative}
\frac{\mathrm d}{{\mathrm d}\alpha}\bar D_\alpha(\rho,\sigma)\big|_{\alpha=1}={\displaystyle \tfrac12} V(\rho,\sigma)\,.
\end{align}

Entropic duality also implies an interesting result on the rate of polarization of a CQ channel under repeated convolution, choosing among the two choices uniformly at random. 
Suppose $Y^n$ is a random variable with values in $\{0,1\}^n$, each with the same probability. 
Then $W_{Y^n}$ is a random convolution of $W$ with itself according to the particular sequence $Y^n$. 
Depending on the application, one is interested in the probability that the resulting channel $W_{Y^n}$ is either essentially deterministic, in that $\Hmin(W_{Y^n})\approx 0$, or essentially random, in that $\Hmax(W_{Y^n})\approx 1$.
Here, and in the remainder of this section, we take the base of the logarithm to be 2.  
The former case is useful in constructing codes for noisy channel communication~\cite{arikan_channel_2009} or information reconciliation~\cite{arikan_source_2010}, the latter for lossy compression~\cite{korada_polar_2010-1} or wiretap coding~\cite{mahdavifar_achieving_2011}.
The rate of polarization refers to how fast the min-entropy approaches 0 with increasing $n$ or how fast the max-entropy approaches $1$, and in principle the rate of polarization to determinstic channels could be distinct from the rate of polarization to random channels.  
However, combining Corollary~\ref{cor:convolutionduality} with Theorem~\ref{thm:eurchannel} implies that the rates must be identical. 
Thus, it is only necessary to establish the precise rate for only one of them. 
This is formalized in the following corollary. 
\begin{corollary}
\label{cor:polarizationrate}
Let $W$ be any symmetric binary-input CQ channel. 
For any function $f$, the following are equivalent:
\begin{align}
\label{eq:limitconvolutionmin}
&\lim_{n\to \infty}P[\Hmin(W_{Y^n})\leq f(n)]=I(W)\qquad \text{and}\\
&\lim_{n\to \infty}P[\Hmax(W_{Y^n})\geq 1-f(n)]=1-I(W)\,.\label{eq:limitconvolutionmax}
\end{align}
Similarly, for $I(W)<1$ and any function $g$, the following are equivalent:
\begin{align}
&\lim_{n\to \infty}P[\Hmin(W_{Y^n})\geq g(n)]=1\qquad \text{and}\\
&\lim_{n\to \infty}P[\Hmax(W_{Y^n})\leq 1-g(n)]=1\,.
\end{align}

\end{corollary}
To see this, first apply \eqref{eq:limitconvolutionmin} to the dual channel and use \eqref{eq:capacityduality} to obtain $\lim_{n\to \infty} P[\Hmin((W^\perp)_{\bar Y^n})\leq f(n)]=1-I(W)$. 
By Theorem~\ref{thm:eurchannel}, the $y^n$ such that $\Hmin((W^\perp)_{\bar y^n})\leq f(n)$ are precisely those for which $\Hmax(W_{y^n})\geq 1-f(n)$. 
This implies \eqref{eq:limitconvolutionmax}. 
The other implications proceed similarly

Note that  polarization statements are not typically made in terms of the min- or max-entropies, but in terms of the Bhattacharyya parameter, which in the quantum case is the fidelity of the output states $B(W):=F(W(0),W(1))$. 
Following the approach of \cite{arikan_rate_2009}, in \cite{wilde_polar_2013} it is shown that $\lim_{n\to \infty}P[B(W_{Y^n})\leq 2^{-2^{n\beta}}]=I(W)$ for any $\beta<\tfrac12$ and any CQ channel $W$. 
Conversely, for any $\beta>\tfrac12$ we have  $\lim_{n\to \infty}P[B(W_{Y^n})\geq 2^{-2^{n\beta}}]=1$. 

We can relate this fidelity to the min-entropy using Lemma 6 of \cite{renes_efficient_2015}. 
This gives $P(W)\geq 1-\tfrac12B(W)$ and therefore $\Hmin(W)\leq -\log (1-\tfrac12 B(W))$, which we can further bound by $B(W)$ itself, since it takes values in $[0,1]$.  
Hence $f(n)=2^{-2^{n\beta}}$ for $\beta<\tfrac12$ is feasible in Corollary~\ref{cor:polarizationrate}. 
On the other hand, the other bound in Lemma 6 yields $P(W)\leq 1 - \tfrac12 (1 - \sqrt{1 - B(W)^2})$, from which the bound $\Hmin(W)\geq \tfrac14 B(W)^2$ follows.
Thus, $g(n)=2^{-2(2^{n\beta}-1)}$ for any $\beta>\tfrac12$ is feasible in the converse statement. 
A more refined analysis would presumably show that $g(n)$ has the same form as $f(n)$, but $\beta>\tfrac12$, but this is left for future work. 

Using an uncertainty relation developed for channel fidelities $B(W)$ leads to a version of Corollary~\ref{cor:polarizationrate} directly in terms of the Bhattacharyya parameter. 
Proposition 3.6 of \cite{renes_alignment_2016} shows that $B(W)+B(W^\perp)\geq 1$. Thus, $P[B(W_{Y^n})\geq 1-f(n)]\geq P[B((W^\perp)_{\bar Y^n})\leq f(n)]$, which implies 
\begin{align}
\lim_{n\to \infty}P[B(W_{Y^n})\geq 1-f(n)]\geq 1-I(W)\,.
\end{align}
A corresponding upper bound follows from the converse bounds on randomness extraction, since exceeding $1-I(W)$  would give a means of extracting random bits from $Z^n$ in $\psi^{\otimes n}_{ZB}$ which are uncorrelated from $B^n$  at a rate greater than $H(Z|B)_\psi$. 
Since the fidelity uncertainty relation is an inequality, it can only be used to show that the rate of polarization to deterministic channels implies a polarization rate to random channels, not vice versa, as in Corollary~\ref{cor:polarizationrate}.

Lemma~\ref{lem:uncertaintyequality} can be directly applied to source scenarios of data compression and randomness extraction, as well as to channel scenarios. 
Importantly, here we are freed from the constraint of symmetry channels and the four corollaries above can be applied to a general state $\ket{\psi}_{ABCD}=\sum_z \sqrt{P_Z(z)}\ket z_A\ket z_C\ket{\varphi_z}_{BD}$ for arbitrary probability distribution $P_Z$ and states $\ket{\varphi_z}$. 
That is, we need not take $P_Z$ to be uniform, as in \eqref{eq:channelstate}, though note that $X_A$ is uniform no matter the choice of $P_Z$. 
Here we have $P(Z_A|B)_\psi=Q(X_A|CD)_\psi$, $Q(Z_A|B)_\psi=P(X_A|CD)_\psi$ from choosing $\Hmin$ and $\Hmax$, $H(Z_A|B)_\psi+H(X_A|CD)_\psi=\log |A|$ from the von Neumann entropy. 
The dispersion argument goes through as above, so that $V(Z_A|B)_\psi=V(X_A|CD)_\psi$. 
Finally, the connection between the rates of polarization also goes through, now using $B(W)=2\sqrt{p_0p_1}F(W(0),W(1))$, and the bound $f(n)=2^{-2^{n\beta}}$ for the general CQ scenario can be obtained from \cite{renes_efficient_2015} following \cite{arikan_rate_2009,arikan_source_2010,wilde_polar_2013}.

\section{Codes and channels}
\label{sec:codechannel}

\subsection{Codes and complementarity}

The notion of duality extends to include linear codes, because a linear code $C$ and its dual $C^\perp$ can be combined into a single quantum code. 
Here we elucidate this combination by taking a somewhat nonstandard approach to describing a code and its dual. 
First consider a reversible linear transformation $M$ from $\mathbb F_q^n$ to itself, with prime $q$. 
We can regard the first $n-k$ outputs of $M$ as defining the parity checks of a linear code $C$ and the remaining $k$ outputs as specifying its encoded information. 
That is, if we define the $(n-k)\times n$ matrix $\hat M$ as the first $n-k$ rows of $M$ and similarly $\bar M$ as the last $k$ rows, then $\hat M$ is the parity check matrix of $C$ and $\bar M$ correspond to the logical bits (message bits). 
Here we regard the matrix as implementing the linear transformation by acting to the right. 
Now let $M'=(M^{-1})^T$ and define $\hat M'$ and $\bar M'$ to be its first $n-k$ and last $k$ rows, respectively. 
Since $M (M')^T=\id$,  $\hat M \bar M'^T=0$, and therefore $\bar M'$ is the parity check matrix of the dual code $C^\perp$. 
We will also have occasion to make use of the code $C^\top$, whose parity check matrix is $\bar M$. 
This is the complement of $C$ in $\mathbb F_q^n$ in the sense that $\mathbb F_q^n=C\oplus C^\top$.
The dual also has a complement, call it $C^{{\topbot}}$, with parity check matrix $\hat M'$. Invertibility of $M$ also implies $\bar M\hat M'^T=0$, meaning $(C^\top)^\perp=C^{{\topbot}}$.

We can promote $M$ to a unitary operator $U$ by using $M$ on the standard basis: $U=\sum_{z^n} \ket{Mz^n}\bra{z^n}$. 
The resulting $U$ then necessarily has the action $M'$ in the conjugate basis: $U\ket{{x^n}}=\ket{(M^{-1})^Tx^n}$ (here and subsequently we drop the tilde and always use $z$  refer to the standard basis or $x$ to the conjugate basis).  
To see this, just use the Fourier transform:
\begin{subequations}
\begin{align}
U\ket{x^n}
&=\tfrac1{\sqrt{q^n}}\sum_{z^n}\omega^{x^n\cdot z^n}\ket{Mz^n}\\
&=\tfrac1{\sqrt{q^n}}\sum_{z^n}\omega^{x^n\cdot M^{-1}z^n}\ket{z^n}\\
&=\tfrac1{\sqrt{q^n}}\sum_{z^n}\omega^{(M^{-1})^Tx^n\cdot z^n}\ket{z^n}\\
&=\ket{(M^{-1})^Tx^n}\,.
\end{align}
\end{subequations}

\subsection{Encoded channel outputs by measurement}
\label{sec:channelcodefromstate}
Just as the outputs of a CQ channel $W$ and its dual can be generated by measuring an appropriate state, the same is true for the encoded outputs of $W^n$. 
However, there are additional subtleties in the encoded case that are worth exploring first before examining duality.

Suppose we apply $U$ to the $A$ systems in $\ket{\psi}^{\otimes n}$, using $\ket{\psi}$ from \eqref{eq:channelstate}. 
Denoting these collectively as $A^n$, the decomposition of $M$ is mirrored in a similar decomposition of $A^n$ into the first $n-k$ and last $k$ systems, call them $\hat A$ and $\bar A$, respectively. 
If we call $E_C$ the encoder of the code $C$, then from the state 
\begin{align}
\label{eq:Psi}
\ket{\Psi}_{\hat A\bar AB^nC^nD^n}:=U_{\hat A\bar A|A^n}\ket{\psi}^{\otimes n}_{ABCD}
\end{align}
we can generate both the outputs of $W^n\circ E_C$ as well as $(W^\perp)^n\circ E_{C^\perp}$.
For $\bar z\in \mathbb F_q^{k}$, we have 
\begin{align}
[W^n\circ E_C](\bar z)=\tr_{C^nD^n}[(\Pi_0)_{\hat A}\otimes \ketbra{\bar z}_{\bar A}\Psi_{\hat A \bar A B^nC^nD^n}]\,.
\end{align}
The projection $\Pi_0$ on $\hat A$ is just $\ketbra 0^{\otimes n-k}$ and ensures that all the parity checks are satisfied and the projection onto $\ket{\bar z}_{\bar A}$ picks out the term in the superposition corresponding to the input $\bar z$. 
Similarly, for $\hat x\in \mathbb F_q^{n-k}$,
\begin{align}
[(W^\perp)^n\circ E_{C^\perp}](\hat x)=\tr_{B^n}[\ketbra{-\hat x}_{\hat A}\otimes (\tilde \Pi_0)_{\bar A}\Psi_{\hat A \bar A B^nC^nD^n}]\,.
\end{align}
Now the projection $\tilde \Pi_0=\ketbra{\tilde 0}^{\otimes k}$ on $\bar A$ ensures that the parity checks of the dual code are satisfied. 
By swapping the roles of the $\hat A$ and $\bar A$ systems, we can equally-well generate the outputs of $W^n\circ E_{C^\top}$ as well as $(W^\perp)^n\circ E_{C^{\topbot}}$. 

For symmetric channels, as we are considering here, we need not insist on using setting all the parity checks to zero as opposed to some other value. 
Put differently, any parity check matrix specifies an entire family of codes, one for each choice of the parity check values (syndromes), and for symmetric channels all codes lead to equivalent decoding tasks.  
To see this more formally, define $E_C(\hat z,\bar z)=M^{-1}(\hat z\oplus \bar z)$, so that the usual encoder is $E_C(\bar z)=E_C(0,\bar z)$. 
Now let $s=M^{-1}(\hat z\oplus 0)$. 
By linearity $E_C(\hat z,\bar z)=s+ E_C(\bar z)$, and therefore by channel symmetry there exists an appropriate unitary operator $V_s$ such that 
\begin{align}
W^n(s+E_C(\bar z))=V_s W^n(E_C(\bar z)) V_s^*\,.
\end{align} 
Hence the two channels $W^n\circ E_C(0,\cdot)$ and $W^n\circ E_C(\hat z,\cdot)$ are equivalent. 
Moreover, we can allow the syndrome to be chosen randomly, provided that it is also delivered as part of the channel output. 
That is, $W^n\circ E_C$ is equivalent to the channel $W'$ which takes $\bar z$ to the pair $(W^n\circ E_C(\hat Z,\bar z),\hat Z)$  for random $\hat Z$.  
This implies that $\Hk(W^n\circ E_C)=\Hk(\bar Z|B^n\hat Z)_\Psi$ for all $\hat z$. 
The latter conditional entropy is relevant in the setting of data compression of $Z^n$, where the decompressor will have access to $B^n$ as well as the compressed output $\hat Z$. 
In particular, $\Hmin(\bar Z|B^n\hat Z)_\Psi$ characterizes the error probability of the compression task, just as it does for the coding task.
In this sense the coding and compression tasks are equivalent for symmetric channels when using linear codes. 

Besides deterministic encoding, it is sometimes useful to employ randomized encoding in which the message is fixed but the syndrome is chosen uniformly at random. 
This is particularly relevant in coding for the wiretap channel, i.e.\ private classical communication. 
Randomized encoding will also play an important role in duality. 
More formally, let $R_C(\bar z)=\tfrac 1{q^{n-k}}\sum_{\hat z} E_C(\hat z,\bar z)$, where the summation is to be understood as the probabilistic mixture of the outputs of $E_C(\hat z,\bar z)$. 
Since the syndrome is unknown at the channel output, the relevant conditional entropy is $\Hk(W^n\circ R_C)=\Hk(\bar Z|B^n)_{\Psi}$. 
This conditional entropy is also relevant in the setting of randomness extraction from $Z^n$ relative to side information $B^n$, where now $\bar Z$ is the output of the extraction scheme.\footnote{Note that we are considering extraction in the case where the underlying joint quantum state is known; in the cryptographic setting one usually considers that only a bound on the min-entropy is known.} 
In particular, $\Hmax(\bar Z|B^n)_{\Psi}$ directly characterizes the closeness of the output to the ideal output, a uniformly random string uncorrelated with system $B^n$. 
As above, randomized encoding and randomness extraction are in this sense equivalent for symmetric channels when using linear codes.  

\subsection{Duality of deterministic and randomized encoding}

Observe that $(W^n\circ E_C)^\perp\nsimeq (W^\perp)^n\circ E_{C^\perp}$ since the input spaces do not match.  
Nor is it the case that $(W^n\circ E_C)^\perp$ is equal or equivalent to $(W^\perp)^n\circ E_{C^{\topbot}}$. 
The latter would require projecting $\hat A$ in $\ket{\Psi}$ onto the conjugate basis state $\ket{\tilde 0}^{\otimes n-k}$, but the construction of $W^n\circ E_C$ uses the projection onto the standard basis state $\ket 0^{\otimes n-k}$.
Instead, the dual converts deterministic encoding to randomized encoding, and vice versa, 
as formalized in the following. 
\begin{prop}
\label{prop:dualencoder}
For any CQ channel $W$ and linear code $C$, 
\begin{align}
(W^n\circ E_C)^\perp&\simeq(W^\perp)^n\circ R_{C^{\topbot}}\quad\text{and}\label{eq:dettorand}\\
(W^n\circ R_C)^\perp&\simeq(W^\perp)^n\circ E_{C^{\topbot}}\,.\label{eq:randtodet}
\end{align}
\end{prop}
\begin{proof}
Consider the state $\ket{\Psi'}_{\hat A\bar AB^n\hat C\bar CD^n}=U_{\hat C\bar C|C^n}\ket{\Psi}_{\hat A\bar AB^nC^nD^n}$, which can be used to compute the two duals in question. More explicitly, we have  
\begin{align}
\label{eq:codebasestate}
\ket{\Psi'}_{\hat A\bar AB^n\hat C\bar CD^n}
=\tfrac1{\sqrt{q^n}}\sum_{\hat z,\bar z}\ket{\hat z}_{\hat A}\ket{\bar z}_{\bar A}\ket{\hat z}_{\hat C}\ket{\bar z}_{\bar C}|\varphi_{M^{-1}(\hat z\oplus \bar z)}\rangle_{B^nD^n}\,.
\end{align}

We first prove \eqref{eq:dettorand}. 
Nominally, the outputs of $W^n\circ E_C$ are obtained by projecting onto $\ket{0^{n-k}}_{\hat A}\ket{\bar z}_{\bar A}$ and keeping the $B^n$ system. 
As described above, we can just as well consider the equivalent scenario in which the output is given by projecting onto $\ket{\bar z}_{\bar A}$ and keeping the $\hat A$ and $B^n$ systems. 
There is no need to measure $\hat A$ to remove superpositions between different syndrome values, as these are wiped out when tracing out $C^n$. 
Thus, the outputs of the dual  $(W^n\circ E_C)^\perp$ are obtained by projecting onto $\ket{{-}\bar x}_{\bar A}$ and keeping the $C^nD^n$ systems.
The projection gives
\begin{align}
\label{eq:codedualprojection}
\bracket{-\bar x}{\Psi'}_{\hat A\bar AB^n\hat C\bar CD^n}
&=\tfrac1{\sqrt{q^n}}\sum_{\hat z,\bar z}\bracket{-\bar x}{\bar z}\ket{\hat z}_{\hat A}\ket{\hat z}_{\hat C}\ket{\bar z}_{\bar C}|\varphi_{M^{-1}(\hat z\oplus \bar z)}\rangle_{B^nD^n}\,.
\end{align}
Defining $\ket{\sigma_{\hat z}}_{\bar CB^nD^n}=\tfrac1{\sqrt{q^k}}\sum_{\bar z}\ket{\bar z}_{\bar C}|\varphi_{M^{-1}(\hat z\oplus \bar z)}\rangle_{B^nD^n}$, the dual channel outputs are then given by 
\begin{align}
\label{eq:randomdualoutputs}
[(W^n\circ E_C)^\perp](\bar x)=Z_{\bar C}^{\bar x}\big(\tfrac{1}{q^{n-k}}\sum_{\hat z}\ketbra{\hat z}_{\hat C}\otimes (\sigma_{\hat z})_{\bar C D^n}\big)Z_{\bar C}^{-\bar x}\,,
\end{align}
since tracing out $\hat A$ dephases the $\hat C$ system.

Meanwhile, the outputs of $(W^\perp)^n\circ R_{C^{\topbot}}$ are by definition $\rho_{\bar x}=\tfrac1{q^{n-k}}\sum_{\hat x} (\theta_{M'^{-1}(\hat x\oplus \bar x)})_{C^nD^n}$, where $\theta_{x^n}=\theta_{x_1}\otimes\cdots\otimes \theta_{x_n}$. 
By symmetry of the $\theta_{x_k}$, 
\begin{align}
\theta_{M'^{-1}(\hat x\oplus \bar x)}=(Z^{M'^{-1}(\hat x\oplus \bar x)})_{C^n}(\theta_{0^n})_{C^nD^n}(Z^{M'^{-1}(\hat x\oplus \bar x)})_{C^n}^*\,.
\end{align}
Observe that applying $U_{\hat C\bar C|C^n}$ to $Z^{M'^{-1}(\hat x\oplus \bar x)}_{C^n}$ results in $Z^{\hat x}_{\hat C}\otimes Z^{\bar x}_{\bar C}$, as might be expected:
\begin{subequations}
\begin{align}
U Z^{M'^{-1}(\hat x\oplus \bar x)} U^*
&=\sum_{z^n} \omega^{z^n\cdot M'^{-1}(\hat x\oplus \bar x)}\ket{Mz^n}\bra{Mz^n}\\
&=\sum_{\hat z,\bar z} \omega^{M^{-1}(\hat z\oplus \bar z)\cdot M'^{-1}(\hat x\oplus \bar x)}\ket{\hat z\oplus \bar z}\bra{\hat z\oplus \bar z}\\
&=\sum_{\hat z,\bar z} \omega^{\hat z\cdot \hat x+\bar z\cdot \bar x}\ketbra{\hat z}\otimes \ketbra{\bar z}\\
&=Z^{\hat x}\otimes Z^{\bar x}
\end{align}
\end{subequations}
Applying $U_{\hat C\bar C|C^n}$ to the $\rho_{\bar x}$ yields an equivalent set of outputs, namely
\begin{align}
\label{eq:randomdualequiv}
U_{\hat C\bar C|C^n}(\rho_{\bar x})_{C^nD^n} U_{\hat C\bar C|C^n}^*
&=\tfrac1{q^{n-k}}\sum_{\hat x} Z_{\hat C}^{\hat x}\otimes Z_{\bar C}^{\bar x} (\theta'_{0^n})_{\hat C\bar CD^n}Z_{\hat C}^{-\hat x}\otimes Z_{\bar C}^{-\bar x}\,,
\end{align}
where we have used $\ket{\theta'_{0^n}}=U_{\hat C\bar C|C^n}\ket{\theta_0}^{\otimes n}$. 
More explicitly,
\begin{subequations}
\begin{align}
\ket{\theta'_{0^n}}_{\hat C\bar CB^nD^n}
&=\tfrac1{\sqrt{q^n}}\sum_{\hat z,\bar z}\ket{\hat z}_{\hat C}\ket{\bar z}_{\bar C}|\varphi_{M^{-1}(\hat z\oplus \bar z)}\rangle_{B^nD^n}\\
&=\tfrac1{\sqrt{q^{n-k}}}\sum_{\hat z}\ket{\hat z}_{\hat C}\ket{\sigma_{\hat z}}_{\bar CB^nD^n}\,.
\end{align}
\end{subequations}
The average over $\hat x$ in \eqref{eq:randomdualequiv} will dephase the $\hat C$ system, leading to equivalent output states identical to those in \eqref{eq:randomdualoutputs}. 

To establish \eqref{eq:randtodet}, first observe that the outputs of $(W^n\circ R_C)$ can be generated by measuring $\bar A$ of $\ket{\Psi'}$ in the standard basis and keeping just the $B^n$ systems. 
Therefore the dual outputs are obtained by measuring $\bar A$ in the conjugate basis and keeping the $\hat A\hat C\bar CD^n$ systems.
The projection is precisely that of \eqref{eq:codedualprojection}, but now since the dual output also includes $\hat A$, it can clearly be absorbed into $\hat C$. 
The dual outputs are then just
\begin{align}
[(W^n\circ R_C)^\perp](\bar x)
&\simeq Z^{\bar x}_{\bar C} \big(\theta'_{0^n}\big)_{\hat C\bar C D^n}Z^{-\bar x}_{\bar C}\,.
\end{align}
Note that this differs from \eqref{eq:randomdualoutputs} in that $\hat C$ is not dephased. 
The outputs of $(W^n)^\perp\circ E_{C^\topbot}$ are just $\theta_{M'^{-1}(0^{n-k}\oplus \bar x)}$. 
By the calculations above for $(W^\perp)^n\circ R_{C^\topbot}$, these are plainly equivalent to $[(W^n\circ R_C)^\perp](\bar x)$. 
\end{proof}

\subsection{
Entropic relations for codes and channels}
\label{sec:entropycodechannel}

There are entropic relationships between channels and codes just as there are for bare channels as in Theorem~\ref{thm:eurchannel}. 
In particular, taking the base of the logarithm to be $q$, we have
\begin{theorem}
\label{thm:eurchannelcode}
For $\ket{\Psi}$ as in \eqref{eq:Psi},  
\begin{align}
&\Hk(\bar Z|B^n\hat Z)_\Psi+\dualHk(\bar X|C^nD^n)_\Psi=k\qquad \text{and}\\
&\Hk(\hat Z|B^n)_\Psi+\dualHk(\hat X|C^nD^n\bar X)_\Psi=n-k\,.
\end{align}
Then, for any CQ channel $W$ and linear code $C$, 
\begin{align}
&\Hk(W^n\circ E_C)+\dualHk((W^\perp)^n\circ R_{C^{\topbot}})=k\qquad\text{and}\\
&\Hk(W^n\circ R_{C^\top})+\dualHk((W^\perp)^n\circ E_{C^\perp})=n-k\,.
\end{align}
\end{theorem}
\begin{proof}
The latter two follow from the former by the discussion in \S\ref{sec:channelcodefromstate}. 
To establish the former, first observe that the following two statements follow from Lemma~\ref{lem:uncertaintyequality}: 
\begin{align}
\Hk(\bar Z|B^n\hat Z)_\Psi+\dualHk(\bar X|C^nD^n\hat Z)_\Psi&=k\qquad \text{and}\\
\Hk(\hat Z|B^n\bar X)_\Psi+\dualHk(\hat X|C^nD^n\bar X)_\Psi&=n-k\,,
\end{align}
Compared to the statements we are trying to prove, here the entropies in the second and first terms are additionally conditioned on $\hat Z$ in the first equation and $\bar X$ in the second, respectively. 
This conditioning can be obtained by extending $\hat A$ to two copies and conditioning on the first copy in the first term and the second copy in the second. 
That is, if we define $\ket{\Psi'}_{\hat A_1\hat A_2\bar AB^nC^nD^n}$ by $\ket{\Psi'}=U_{\hat A_1\hat A_2|\hat A}\ket{\Psi}$ for $U_{\hat A_1\hat A_2|\hat A}=\sum_{\hat z}\ket{\hat z}_{\hat A_1}\ket{\hat z}_{\hat A_2}\bra{\hat z}_{\hat A}$, then  
$\Hk(\bar Z|B^n\hat Z)_\Psi=\Hk(\bar Z|B^n\hat A_1)_{\Psi'}$ and $\dualHk(\bar X|C^nD^n\hat Z)_\Psi=\dualHk(\bar X|C^nD^n\hat A_2)_{\Psi'}$. 
Applying Lemma~\ref{lem:uncertaintyequality} to $\ket{\Psi'}$ with $A$ therein equal to $\bar A$ here, $E=\hat A_1B^n$ and $F=\hat A_2C^nD^n$ gives the first equality, and an entirely similar argument gives the second. 

It then remains to show that $\hat Z$ is irrelevant in the second term of the first equation and $\bar X$ is irrelevant in the first term of the second. 
We can dispense with $\hat Z$ in $\dualHk(\bar X|C^nD^n\hat Z)_\Psi$ since it can be obtained from $C^n$ anyway; it is redundant. 
On the other hand, we can dispense with $\bar X$ in $\Hk(\hat Z|B^n\bar X)_\Psi$ because tracing out $C^nD^n$ leaves the $\bar A$ system of $\Psi_{\hat A\bar AB^n}$ in a random $\ketbra{\bar z}$ state. 
Thus measurement of $\bar X$ results in a random outcome, completely independent of the remaining parts of $\Psi$. 
\end{proof} 

Due to the connections with the source tasks of data compression and randomness extraction, Theorem~\ref{thm:eurchannelcode} allows us to convert randomness extractors for symmetric sources into error-correcting codes for symmetric channels and vice versa. 
This was first suggested by the author in \cite{renes_duality_2011}, but here we can draw much tighter conclusions. 
Note that this is a different relation between codes and extractors than that of e.g.\ Ta-Shma and Zuckerman~\cite{ta-shma_extractor_2004}. 
Even the setting therein is different; while here we consider extraction from known sources, whereas randomness extraction in the cryptographic literature usually refers to functions which produce randomness from sources that are only guaranteed to have a certain min-entropy.
The resulting extractor codes have codewords which are sequences running through the different seed values.
This has no analog in the present setting, as there is no seed. 
Instead, for a length-$n$ code $C$ encoding length-$k$ messages, the corresponding randomness extractor function is given by $f(x^n)=\bar M' x^n$. 
Observe that $\bar M'$ is the generator matrix of the code, but is used in the opposite sense by the extractor; messages $m^k$ are encoded as $z^n=m^k \bar M'$. 

As a simple example of the use of Theorem~\ref{thm:eurchannelcode}, consider the recent result that Reed-Muller codes achieve the capacity of the binary erasure channel~\cite{kudekar_reed-muller_2016-1}.
By duality, this implies that Reed-Muller codes can also extract randomness at the optimal rate from the source describing the joint input and output to the BEC. 
If $\hat M$ is the parity-check matrix of a Reed-Muller code used for error-correction, then the associated extractor function is given by the matrix $\bar M'$ acting to the right. 
This is the parity check matrix of the dual code, which is also a Reed-Muller code. 
The error of the optimal decoder will translate directly into the quality of the extracted randomness, which follows by choosing $\Hk=\Hmin$ in Theorem~\ref{thm:eurchannelcode} just as in Corollary~\ref{cor:PQ}. 
The rates of the two procedures are of course also linked by the relationship between $\hat M$ and $\bar M'$; since the extractor uses the generator matrix of the code, the size of the extractor output is just the size of the code $|C|$. 
By self-duality of the BEC, a code of size $|C|$ with error $\eps$ for $\text{BEC}(p)$ also functions as a randomness extractor of output length $|C|$ and quality $Q=1-\eps$ from the source describing $\text{BEC}(1-p)$. In the limit of large blocklength, $|C|$ will tend to $1-p$ while $\eps$ tends to zero. 

This argument could just as well be run in reverse to convert an optimal-rate extractor into a capacity-achieving channel code. 
It would be interesting to further investigate this possibility, for instance to construct an optimal extractor for the dual of the BSC and thereby find a capacity-achieving code for the binary symmetric channel. 
Indeed, by the arguments in Appendix 7 of \cite{sasoglu_polar_2011}, this would provide a capacity-achieving code for \emph{all} symmetric binary-input channels.

Just as remarked at the end of \S\ref{sec:entropychannel}, the first part of Theorem~\ref{thm:eurchannelcode} also holds for $\ket{\Psi}=U_{\hat A\bar A|A}\ket{\psi}^{\otimes n}$ with more general $\ket{\psi}_{ABCD}=\sum_z \sqrt{P_Z(z)}\ket z_A\ket z_C\ket{\varphi_z}_{BD}$. 
Thus, not only can we relate entropic properties of the sources involved in data compression and randomness extraction, but also protocols for the two. 
It turns out that the sizes of the optimal data compression and randomness extraction procedures (using linear codes) satisfy a simple relationship for any blocklength $n$. 
As in \cite{tomamichel_hierarchy_2013}, let $m^L_\eps(Z_A|B)_\psi$ be the minimal compression length of $Z^n$ relative to $B^n$ using a linear code with error $\eps$. 
By the discussion in \S\ref{sec:channelcodefromstate}, this is the smallest $|\hat Z|$ such that $P(\bar Z|B^n\hat Z)_{\Psi}\geq 1-\eps$, for $\bar Z$ and $\hat Z$ obtained from code $C$. 
Similarly, following \cite{tomamichel_hierarchy_2013}, let $\ell^L_\eps(X_A|CD)_\psi$ be the maximal randomness extractable by a linear function from $X^n$, which is independent of $C^nD^n$ up to quality $1-\eps^2$. 
Again by \S\ref{sec:channelcodefromstate}, this is the largest $|\bar X|$ such that $Q(\bar X|C^nD^n)_\Psi\geq 1-\eps^2$. 
This formulation ensures that the purification distance between the actual and ideal outputs is less than $\eps$. 
Then we have
\begin{corollary}
\label{cor:finiteblock}
For any state $\ket\psi_{ABCD}=\sum_z \sqrt{P_Z(z)}\ket{z}_A\ket z_C\ket{\varphi_z}_{BD}$ and any $\eps\in [0,1]$,
\begin{align}
m^L_{\eps^2}(Z_A|B)_\psi+\ell^L_{\eps}(X_A|CD)_\psi=n\,.
\end{align}
\end{corollary}
\begin{proof}
First note that by picking $\Hk=\Hmin$ in Theorem~\ref{thm:eurchannelcode}, we have $P(\bar Z|B^n\hat Z)_\Psi=Q(\bar X|C^nD^n)_\Psi$. 
Now suppose the optimal linear compression procedure with error $\eps^2$ is based on the linear transformation of $Z^n$ to compressed output $\hat Z$, so that $P(\bar Z|B^n\hat Z)_\Psi\geq 1-\eps^2$. 
Thus the transformation of $X^n$ to $\bar X$ has $Q(\bar X|C^nD^n)_\Psi\geq 1-\eps^2$, meaning  $\ell^L_\eps(X_A|CD)_\psi\geq n-\log |\hat Z|=n-m_{\eps^2}(Z_A|B)_\psi$. 
For the opposite bound, suppose the optimal linear extraction procedure with parameter $\eps$ uses the transformation of $X^n$ to $\bar X$, meaning $Q(\bar X|C^nD^n)_\Psi\geq 1-\eps^2$. 
The transformation from $Z^n$ to $\hat Z$ satisfies $P(\bar Z|B^n\hat Z)_\Psi\geq 1-\eps^2$, implying $m_\eps^L(Z_A|B)_\psi\leq n-\log |\bar X|=n-\ell_{\eps^2}^L(X_A|CD)_\psi$.  
\end{proof}

Hence bounds on randomness extraction can be applied to data compression, and vice versa, at least for compression and extraction based on linear codes. 
For example, this gives a unified derivation of the second-order asymptotic analysis of these tasks in \cite{tomamichel_hierarchy_2013}. 
Starting from Corollary 15 therein, $m_{\eps^2}^L(Z_A|B)_\psi=n H(Z_A|B)_\psi+\sqrt{n}V(Z_A|B)_\psi\Phi^{-1}(1-\eps^2)+O(\log n)$ we immediately have $\ell^L_\eps(X_A|CD)_\psi = n(\log |A|-H(Z_A|B)_\psi)-\sqrt{n}V(Z_A|B)_\psi\Phi^{-1}(1-\eps^2)+O(\log n)$, which upon using the relations $H(Z_A|B)_\psi+H(X_A|CD)_\psi=\log |A|$ and $V(Z_A|B)_\psi=V(X_A|CD)_\psi$ is Corollary 16. 
Similarly, we could start from Corollary 16 and infer Corollary 15. 
Note that the restriction to linear compression and extraction schemes does not affect this argument, since the converse of each applies to linear schemes and the achievability statements are established using two-universal hashing, which includes linear schemes. 

An interesting question is how the bounds on compression and extraction compare for finite blocklength. 
For example, we can significantly tighten the bounds on randomness extraction from the CQ ensemble considered in \cite{tomamichel_hierarchy_2013}. 
They consider the state 
\begin{align}
\label{eq:extractorstate}
\psi_{XC}=\tfrac12\sum_{x}\ketbra x_X\otimes (Z^x \ketbra{\eta} Z^x)_C,
\end{align}
where $\ket{\eta}=\sqrt{p}\ket{0}+\sqrt{1-p}\ket{1}$.
This is precisely the output of the dual to the binary symmetric channel, and since compression is directly related to coding for symmetric channels, we can use bounds on the coding problem to infer bounds on randomness extraction from $\psi^{\otimes n}$. 
In particular, the metaconverse involving the hypothesis-testing quantity $\beta_\eps$ applies to linear codes~\cite{nagaoka_strong_2001,polyanskiy_channel_2010-1}\cite[Lemma 4.7]{hayashi_quantum_2017}, as does the Poltyrev achievability bound~\cite{poltyrev_bounds_1994}. 
For versions of these bounds specifically formulated for the BSC, see Theorems 34 and Theorem 35 of \cite{polyanskiy_channel_2010-1}.
Alternately, the achievability bound of Theorem 9 in \cite{tomamichel_hierarchy_2013} also involves $\beta_\eps$, making it easier to compute though significantly worse than the Poltyrev bound. 
A comparison of the bounds appears in Figure~\ref{fig:plot}.
One could also investigate the comparison in the other direction, using extraction bounds for classical side information generated by the BSC to give bounds on the coding problem for the channel with pure state outputs. A converse for extraction in classical scenarios is formulated in~\cite[Lemma 19]{hayashi_secret_2016}, for instance.

\begin{figure}
\centering
\includegraphics{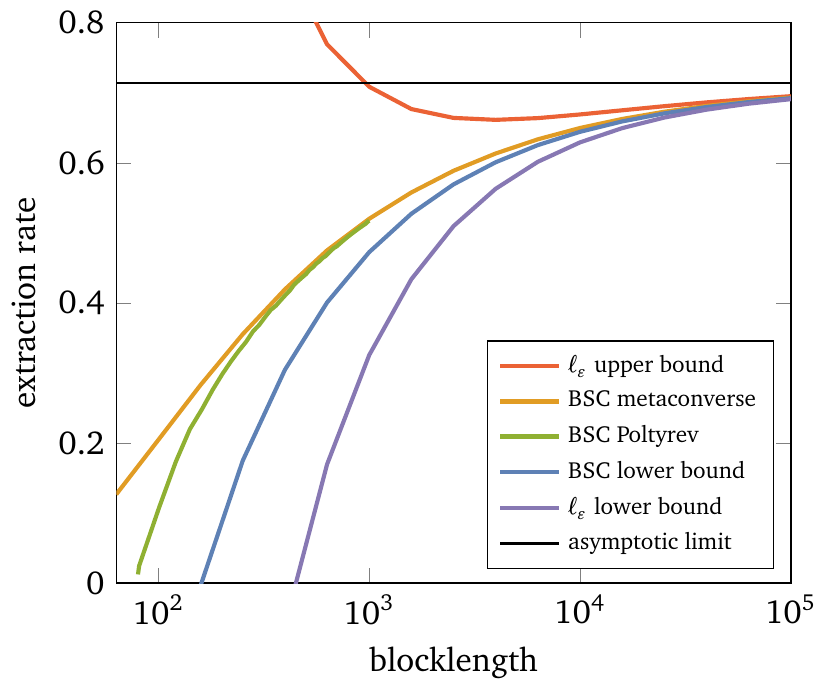}
\caption{\label{fig:plot} Comparison of finite blocklength bounds on randomness extraction from the CQ state in \eqref{eq:extractorstate}. 
The $\ell_{\!\eps}$ upper and lower bounds 
are based on information-spectrum quantities. 
By duality, specifically Corollary~\ref{cor:finiteblock}, tighter bounds are available by appealing to bounds on coding rates for the binary symmetric channel. 
The BSC metaconverse and lower bound 
are both based on the hypothesis testing quantity $\beta_\eps$. 
The Poltyrev achievability bound, 
based on weight spectra of linear codes, already matches the metaconverse extremely closely for blocklengths in the hundreds: At 500 the bounds differ by just four bits! 
However, it becomes time-consuming to compute for blocklengths in the thousands.}
\end{figure}

Not every extraction problem will be the dual of such a simple classical channel, particularly not one for which very tight bounds can be readily computed. 
Tomamichel and Hayashi refer to bounds involving the hypothesis-testing quantity, as we have used in the example, as giving a ``microscopic'' analysis of the coding problem, and hence very tight bounds. 
But in general such a microscopic analysis cannot be easily performed, and instead one has to rely on a more ``macroscopic'' approach, their term for employing information-spectrum quantities (whose definition we shall not give here). 
Indeed, the bounds on extraction in their example are computed using this approach. 
It would be interesting to compare the performance of their macroscopic bounds on compression and extraction, in particular Theorem 17, in light of Corollary~\ref{cor:finiteblock}. 
We leave this question to future work. 

\subsection{EXIT functions}

Duality also implies that the EXIT function of a channel and code combination and that of the dual channel and dual code combination sum to a fixed constant, the logarithm of the alphabet size.
The EXIT function for a code and channel is defined as follows. 
Let $Z^n$ be a random codeword in $C$ and denote by $Z_i$ the $i$th bit of $Z^n$. 
For $B^n=W^n(Z^n)$, denote by $B_{\ssim i}^n$ everything but the $i$th $B$ system. 
Then the EXIT function using entropy $\Hk$ is 
\begin{align}
\label{eq:exitf}
\Xi_{\Hk}(W,C):= \frac1n\sum_{i=1}^n \Hk(Z_i|B_{\ssim i}^n)\,.
\end{align}
Nominally the EXIT function is defined in terms of the von Neumann or Shannon conditional entropy, but here will consider more general $\Hk$ or $\dualHk$ entropies.
For simplicity, we omit the smooth min- and max-entropies and show 
\begin{theorem}
\label{thm:exit}
For any symmetric CQ channel $W$ with input alphabet of size $q$ and linear code $C$, 
\begin{align}
\Xi_\Hk(W,C)+\Xi_{\dualHk}(W^\perp,C^\perp)=\log q\,,
\end{align}
where $\Hk$ is any entropy in \eqref{eq:entropydown} or \eqref{eq:entropyup}.
\end{theorem}
\begin{proof}
By symmetry and the discussion in \S\ref{sec:channelcodefromstate}, it is sufficient to show
\begin{align}
\Hk(Z_i
|B_{\ssim i}^n\hat Z)_{\psi^{\otimes n}}+\dualHk(X_i|C_{\ssim i}^nD_{\ssim i}^n\bar X)_{\psi^{\otimes n}}=\log q\,,
\end{align}
for $\psi$ from \eqref{eq:channelstate} and where $Z_i$ refers to the result of measuring the $Z$ observable of the $i$th bit of $A^n$, $\hat Z$ to the value of the syndrome measurement, and similarly for $X_i$ and $\bar X$. 
Again the goal is to make use of Lemma~\ref{lem:uncertaintyequality}, though doing so requires a little work.

First note that $\hat Z$ can be regarded as a sequence of $Z$-type operators, usually called stabilizers, one for each of the rows of $\hat M$. 
That is, $\hat Z=(\hat Z_1,\dots,\hat Z_{n-k})$, where $\hat Z_j=Z^{\hat M_j}$ and $Z^v=Z^{v_1}\otimes Z^{v_2}\otimes \cdots \otimes Z^{v_n}$.
By employing row reduction, we can assume without loss of generality that $\hat M$ has only one $1$ in the $i$th column. 
This implies that only one of the stabilizers involves the $i$th qubit, and we can also assume without loss of generality that it is the first. 
Then $\hat Z_1=Z_i \cdot Z_1'$, for $Z'_1$ a $Z$-type operator on the remaining $n-1$ qubits. 
Let us denote the set of remaining $n-k-1$ stabilizers $\hat Z_{\ssim 1}$. 
By a similar procedure we can define $\bar X_1=X_i\cdot X_1'$ and $\bar X_{\ssim 1}$. 
The two row reduction procedures are independent, since row reduction does not affect orthogonality.  

Since the stabilizers all commute, but $X_i$ and $Z_i$ anticommute, so too do $X_1'$ and $Z_1'$. 
Now use $\Hk(Z_i|B_{\ssim i}^n\hat Z)_{\psi^{\otimes n}}=\Hk(Z'_1|B^{n-1}\hat Z_{\ssim 1})_{\psi^{\otimes n-1}}$ and $\dualHk(X_i|C^n_{\ssim i}D^n_{\ssim i}\bar X)_{\psi^{\otimes n}}=\dualHk(X'_1|C^{n-1}D^{n-1}\bar X_{\ssim 1})_{\psi^{\otimes n-1}}$ from the following Lemma~\ref{lem:simpleH}. 
Projecting onto fixed values for $\hat Z_{\ssim 1}$ and $\bar X_{\ssim 1}$ yields a pure state, and certainly $Z_1'$ can be obtained by measuring the $C^{n-1}$ appropriately. 
Thus, we may apply Lemma~\ref{lem:uncertaintyequality} to complete the proof.
\end{proof}

\begin{lemma}
\label{lem:simpleH}
For a CQ state of the form $\psi_{XYB}=\tfrac1{|X|}\sum_{xz}P_{Y}(z)\ketbra x_X\otimes \ketbra{x+z}_Y\otimes (\sigma_z)_B$, let $\rho_{YB}=\sum_y P_Y(y)\ketbra{y}_Y\otimes (\sigma_y)_B$. 
Then, for any conditional entropy measure from \eqref{eq:entropydown} or \eqref{eq:entropyup},
\begin{align}
\Hk(X|YB)_\psi&=\Hk(Y|B)_\rho\quad\text{and}\label{eq:remove1}\\
\dualHk(X|YB)_\psi&=\dualHk(Y|B)_\rho\,.\label{eq:remove2}
\end{align}
\end{lemma}

\begin{figure}
\centering
\includegraphics{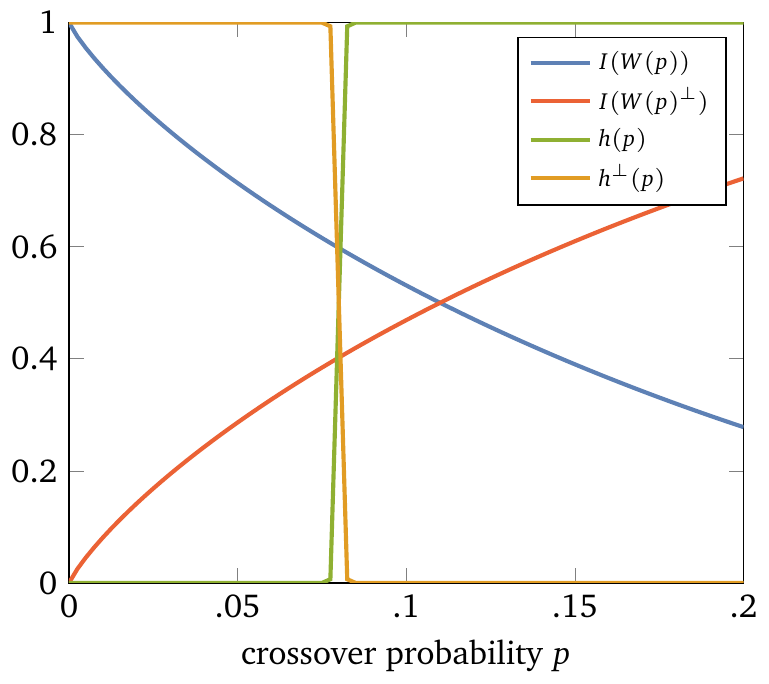}
\caption{\label{fig:exit} EXIT function transition and capacity. 
The figure depicts the capacity of $W(p)$ the BSC with crossover probability $p$, the capacity of its dual $W(p)^\perp$, as well as a putative EXIT function for a rate $R=\nicefrac12$ code over $W(p)$ and EXIT function of the dual code over $W(p)^\perp$. 
Here the EXIT function $h(p)$ displays a sharp transition at $p^\star=0.8$ such that the capacity $I(W(p^\star))\approx 0.6$ exceeds the rate $R$. 
By duality, this implies that the dual code, which also has rate $\nicefrac12$, is reliably decodable for values of $p$ (say $0.1$) such that the rate exceeds the capacity of $W(p)^\perp$ (${\approx}0.47$). 
As this cannot be the case by the converse to the noisy channel coding theorem, it must be that the transition satisfies $I(W(p^\star))=R$, i.e.\ $p^\star\approx 0.11$ in this example.}
\end{figure}

The proof is given in Appendix \ref{sec:simpleH}.
EXIT functions figure prominently in the study of belief propagation decoding~\cite{richardson_modern_2008}, as well as in the recent proof by Kudekar \emph{et al.}\ that Reed-Muller codes achieve capacity on erasure channels~\cite{kudekar_reed-muller_2016-1}.
Let us briefly recall their proof; we will then be able to see how Theorem~\ref{thm:exit} offers a potential route to generalizing the argument for other channels. 
The proof is based on the fact that the $i$th EXIT function (the $i$th term in \eqref{eq:exitf} using the Shannon entropy) is the error probability of the optimal bitwise decoder for the BEC, so that if the EXIT function is essentially zero, then decoding is reliable.
For doubly transitive codes like Reed-Muller codes, the EXIT function is the same for each codeword bit; let us  define $h(p)=\Xi_H(\text{BEC}(p),C)$ as the EXIT function (using the Shannon entropy) for a given code $C$. 
Kudekar \emph{et al.}\ show that for doubly-transitive codes $h(p)$ exhibits a sharp transition as $p$ increases, jumping from zero to one in an interval that decreases with the blocklength. 
The location of the transition depends on the chosen code $C$, and in particular it must not be so high as to imply that the code is reliably decodable above the capacity of the channel. 
For $\text{BEC}(p)$ the capacity is $1-p$, and therefore for given code of rate $R$ the transition $p^\star$ must satisfy $R\leq 1-p^\star$.
The area theorem implies that in fact $p^\star=1-R$, so Reed-Muller codes achieve capacity. 

Theorem~\ref{thm:exit} offers two means of potentially extending this argument to more general chanels. 
First, one can shift the problem of showing a transition in the EXIT function to that of the dual. 
To study the BSC for instance, one could instead look at the EXIT function associated with the state in \eqref{eq:extractorstate}, which may be easier to study with existing tools. 
Moreover, one can examine EXIT functions for different entropies, for example $\Hmin$, and still appeal to duality. 

Secondly, Theorems~\ref{thm:exit} and \ref{thm:eurchannel} imply that, for any channel, if a sharp transition exists, it should be located at capacity.
Here we give a rough sketch of the argument, which is also illustrated in Figure~\ref{fig:exit}. 
Let us simplify to the case of binary input channels, $q=2$, and fix a binary code $C$ for use on the family of channels $W(p)=\text{BSC}(p)$, with $p$ the crossover probability.
Note that, by Theorem~\ref{thm:eurchannel}, the family $W^\perp(p)$ is decreasingly noisy with increasing $p$. 
Defining $h(p)=\Xi_H(W(p),C)$ and $h^\perp(p)=\Xi_H(W(p)^\perp,C^\perp)$, we immediately have $h(p)+h^\perp(p)=1$ by Theorem~\ref{thm:exit}. 
Thus, $h(p)$ has a sharp transition if and only if $h^\perp(p)$ does. 
As before, the transition value $p^\star$ must be constrained by the capacity, else codes with rates exceeding the capacity could still be reliably decoded. (For the strong converse to CQ channel coding, see~\cite{winter_coding_1999}.)
This means we must have $R\leq I(W(p^\star))$ for the family $W(p)$, as well as $R^\perp\leq I(W(p^\star)^\perp)$ for $W(p)^\perp$, where $R^\perp$ is the rate of $C^\perp$.
But $R+R^\perp=1$ by construction, and $I(W(p^\star))+I(W(p^\star)^\perp)=1$ by Theorem~\ref{thm:eurchannel}, which implies that $p^\star$ satisfies $R=I(W(p^\star))$.

\section{Discussion}
\label{sec:discussion}
We have shown that a channel and its dual are very tightly related by uncertainty relations for a general class of entropies, and that this duality is compatible both with channel convolution as in polar coding and with the use of linear codes and their duals. 
We have also investigated several consequences of duality, finding applications to the phenomenon of polarization, the relationship between randomness extraction and coding, as well as possible means of showing a code family achieves capacity of a given channel. 

We have confined much of our analysis to symmetric channels, though this restriction was seen to be unnecessary in applications to source problems such as compression or randomness extraction. 
It would be interesting to extend the results to nonsymmetric channels, but there appear to be some obstacles to doing so. 
The chief difficulty is that we are confined to considering uniform inputs to $W$ in order to employ the definition of the dual using \eqref{eq:channelstate}. 
For example, suppose we take $W$ to be the classical Z channel, for which the capacity-achieving distribution is not uniform. 
In anticipation of using the entropic uncertainty relation, we could consider the state $\ket{\psi}_{ABCD}=\sum_z \sqrt{p_z}\ket{z}_A\ket{z}_C\ket{\varphi_z}_{BD}$ where $p_0$ and $p_1$ are the capacity-achieving distribution and $\varphi_z$ are the Z channel outputs. 
The capacity itself is then $I(W)=H(Z_A)_\psi-H(Z_A|B)_\psi$, and $H(Z_A|B)_\psi+H(X_A|CD)_\psi=\log 2$ certainly still holds. 
However, $H(X_A|CD)_\psi$ is no longer immediately related to $W^\perp$. 
Nevertheless, for questions involving average coding error of $W$, where a uniform input is appropriate, Theorems~\ref{thm:eurchannel} and \ref{thm:eurchannelcode} continue to give meaningful relations to the dual. 

The proof of Theorem~\ref{thm:exit} raises an interesting question regarding entropic uncertainty relations that, to the author's knowledge, has somehow eluded previous investigation.
Therein it is crucial that the $i$th output system be excluded from the conditioning system in order to be able to appeal to Lemma~\ref{lem:uncertaintyequality}, though ultimately the quantity $\Hk(Z_i|B^n\hat Z)$ is perhaps more relevant. 
While the lemma is a novel relation in that it holds with equality, it is still in the usual ``tripartite'' framework of inequality uncertainty relations as discovered in \cite{renes_conjectured_2009,berta_uncertainty_2010,tomamichel_uncertainty_2011,coles_uncertainty_2012}. 
That is, one of two (conjugate) measurements is performed on a system $A$ and we are interested in the conditional entropies $\Hk(X_A|E)$ and $\Hk^\perp(Z_A|F)$ for distinct $E$ and $F$. 
Using chain rules for the von Neumann entropy, one can convert between the tripartite and ``bipartite'' version which is a relation between $H(X_A|E)$ and $H(Z_A|E)$. 
Considering $\Hk(Z_i|B^n\hat Z)$ suggests a different scenario intermediate between bipartite and tripartite, where one is interested in a relation involving $\Hk(X_A|EC)$ and $\Hk^\perp(Z_A|FC)$. 
Whether a useful relation exists is an open question. 

\vspace{2mm}
{\bfseries{Acknowledgments.}} I thank Fr\'ed\'eric Dupuis, S.\ Hamed Hassani, Marco Mondelli, Rajai Nasser, David Sutter, Marco Tomamichel, Rüdiger Urbanke, and Andreas Winter for helpful and interesting discussions. Thanks also to Narayanan Rengaswamy for very carefully reading through an earlier draft of this manuscript. 
This work was supported by the Swiss National Science Foundation (SNSF) via the National Centre of Competence in Research “QSIT”, as well as the Air Force Office of Scientific Research (AFOSR) via grant FA9550-16-1-0245.

\printbibliography[heading=bibintoc,title=References]

\appendix
\section{Entropic uncertainty relations}
\label{app:eur}
In the context of entropy duality, it is more convenient to describe CQ conditional entropies $\Hk(Z_A|E)_\psi$ by using isometries to generate the required state, not directly by measuring $\psi_{AE}$.
Abusing notation somewhat, let  $U_{XA|A}=\sum_x \ket{x}_X\cket{x}\cbra{x}_A$ and $U_{ZA|A}=\sum_z \ket{z}_Z\ket{z}\bra{z}_A$. 
Then defining $\xi_{XAE}=U_{XA|A}\psi_{AE}U_{XA|A}^*$ and $\eta_{ZAE}=U_{ZA|A}\psi_{AE}U_{ZA|A}^*$ we have $\Hk(X_A|E)_\psi=\Hk(X|E)_\xi$ and $\Hk(Z_A|E)_\psi=\Hk(Z|E)_\eta$. 
It will also be useful to define $\Pi_{XA}=\sum_x \ketbra x_X\otimes \ketbra{\tilde x}_A$, which is the projector onto the image of $U_{XA|A}$.  

\begin{proof}[Proof of Lemma~\ref{lem:uncertaintyequality}]
That the sum of entropies is not smaller than $\log |A|$ is the entropic uncertainty relation of conjugate observables, first shown for the von Neumann entropy in~\cite{renes_conjectured_2009}, then for smooth min- and max-entropies in~\cite{tomamichel_uncertainty_2011}, and generally for entropies based on divergence in \cite{coles_uncertainty_2012}. 
Hence we need only establish the upper bound. 

To do so, first write $\ket\psi_{AEF}=\sum_z \sqrt{P_Z(z)}\ket z_A\ket{\sigma_z}_{EF}$ for $\sqrt{P_Z(z)}\ket{\sigma}_{EF}={}_A\bracket{z}{\psi}_{AEF}$. 
The conditional marginals of $F$ are then the reduced states of $\ket{\sigma_z}_{EF}$. Since these are disjoint, the state $\ket{\psi'}_{AEFZ}=\sum_z \sqrt{P_Z(z)}\ket{z}_A\ket{z}_Z\ket{\sigma_z}_{EF}$ can be created locally by measuring $F$. 
Hence $\dualHk(X_A|F)_\psi=\dualHk(X_A|FZ)_{\psi'}$. 
Furthermore, $\psi_{AE}=\psi'_{AE}=\sum_z P_Z(z)\ketbra z_A\otimes (\sigma_z)_E$, and thus $\Hk(Z_A|E)_\psi=\Hk(A|E)_{\psi'}$.
Therefore, we need only show that $\Hk(A|E)_{\psi'}+\dualHk(X_A|F)_{\psi'}\leq \log_2|A|$ for states of this form. 
This is done for the various cases in the following Lemmas \ref{lem:downineq}, \ref{lem:upineq}, \ref{lem:minmax}, and \ref{lem:maxmin}.
\end{proof}
In fact, only the case of $\Hk_\downarrow$ requires the state $\ket{\psi}$ to have this precise form.
\begin{lemma}
\label{lem:downineq}
For $\Hk_\downarrow$ any conditional entropy measure as in \eqref{eq:entropydown} and any normalized pure state $\ket{\psi}_{AEFZ}$ of the form $\ket{\psi}_{AEFZ}=\sum_z \sqrt{P_Z(z)}\ket z_A\ket z_Z\ket{\sigma_z}_{EF}$ with arbitrary $\ket{\sigma_z}$,  
\begin{align}
\Hk_\downarrow(A|E)_\psi+\dualHk_\downarrow(X_A|FZ)_\psi\leq \log |A|\,.
  \end{align}
\end{lemma}
\begin{proof}
Define $\ket{\xi}_{XAEFZ}=\sum_{xz}\bracket{\tilde x}{z}\sqrt{P_Z(z)}\ket{x}_X\cket{x}_A\ket z_Z\ket{\sigma_z}_{EF}$ and observe that $\xi_{AE}=\mu_A\otimes \psi_E$ and $\psi_E=\sum_z P_Z(z)(\sigma_z)_E$
By entropy duality we have $\dualHk_\downarrow(X_A|FZ)_\psi=\dualHk_\downarrow(X|FZ)_\xi=-\Hk_\downarrow(X|AE)_\xi$. 
Then 
\begin{subequations}
\begin{align}
-\Hk_\downarrow(A|E)_\psi
&=\Dk(\psi_{AE},\id_A\otimes \psi_E)\\
&=\Dk(U_{XA|A}\psi_{AE}U_{XA|A}^*,U_{XA|A}U_{XA|A}^*\otimes \psi_E)\label{eq:optverstart}\\
&=\Dk(\xi_{XAE},\Pi_{XA}\otimes \psi_E)\\
&\geq \Dk(\xi_{XAE},\id_{XA}\otimes \psi_E)\\
&=\Dk(\xi_{XAE},\id_X\otimes \mu_A\otimes \psi_E)-\log_2 |A|\label{eq:optverend}\\
&=\Dk(\xi_{XAE},\id_X\otimes \xi_{AE})-\log_2 |A|\\
&=-\dualHk_\downarrow(X_A|FZ)_\psi-\log_2 |A|\,.
\end{align}
\end{subequations}
The first equality is the definition of $\Hk_\downarrow(A|E)_\psi$, the second invariance of $\Dk$ under isometries. 
In the third we use the fact that $UU^*=\Pi$. 
The inequality uses the dominance property of $\Dk$, since, $\Pi\leq \id$, while the following equality uses normalization. 
The penultimate equality uses the specific form of $\xi_{XAE}$, and the final equality the entropy duality relation above.  
\end{proof} 

\begin{lemma}
\label{lem:upineq}
For $\Hk_\uparrow$ any conditional entropy measure as in \eqref{eq:entropyup} and any normalized pure state $\ket{\psi}_{AEF}$,
\begin{align}
\Hk_\uparrow(A|E)_\psi+\dualHk_\uparrow(X_A|F)_\psi\leq \log |A|\,.
  \end{align}
\end{lemma}
\begin{proof}
We have the same entropy duality $\dualHk_\uparrow(X_A|F)_\psi=-\Hk_\uparrow(X|AE)_\xi$, with $\ket{\xi}_{XAEF}=U_{XA|A}\ket{\psi}_{AEF}$. 
Then 
\begin{subequations}
\begin{align}
-\Hk_\uparrow(A|E)_\psi
&=\min_{\tau}\Dk(\psi_{AE},\id_A\otimes \tau_E)\\
&\geq \min_\tau \Dk(\xi_{XAE},\id_X\otimes \mu_A\otimes \tau_E)-\log|A|\\
&\geq \min_\sigma \Dk(\xi_{XAE},\id_X\otimes \sigma_{AE})-\log|A|\\
&=-\dualHk(X_A|F)_\psi-\log_2 |A|\,.
\end{align}
\end{subequations}
The first inequality encapsulates \eqref{eq:optverstart}-\eqref{eq:optverend} from the proof of Lemma~\ref{lem:downineq}, since these steps hold for an arbitrary $\psi_E$.
\end{proof}

\begin{lemma}
  \label{lem:minmax}
  For any pure state $\ket{\psi}_{AEF}$ and $0\leq \eps\leq 1$,  \begin{align}
    &\Hmin^\eps(A|E)_\psi+\Hmax^\eps(X_A|F)_\psi\leq \log |A|\,.
  \end{align}
\end{lemma}
\begin{proof}
  Define $\lambda=H_{\rm min}^\eps(A|E)_\psi$ and let $\widetilde{\psi}_{AE}\in \mathcal B_\eps(\psi_{AE})$ and $\sigma_E$ be such that 
$\widetilde \psi_{AE}\leq 2^{-\lambda}\id_A\otimes \sigma_E.$
Applying the isometry $U_{XA|A}$ yields
\begin{subequations}
\begin{align}
\widetilde \xi_{XAE}&= U_{XA|A}\widetilde\psi_{AE}U^*_{XA|A}\\
&\leq 2^{-\lambda} U_{XA|A}(\id_A\otimes \sigma_E) U_{XA|A}^*\\
&=2^{-\lambda}\Pi_{XA}\otimes \sigma_E\\
& \leq 2^{-\lambda}\id_{XA}\otimes \sigma_E\\
&= 2^{-(\lambda-\log |A|)}\id_X\otimes \mu_A\otimes \sigma_E\,.
\end{align}
\end{subequations}
Note that the first step implies $\widetilde \xi_{XAE}\in\cB_\eps(\xi_{XAE})$. 
Thus, $\nu=\lambda-\log |A|$ and $\mu_A\otimes\sigma_E$ are feasible for $\Hmin^\eps(X|AE)_\xi$, meaning 
\begin{align}
\Hmin^\eps(X|AE)_\xi\geq \Hmin^\eps(A|E)_\psi-\log |A|.
\end{align}
Therefore the claim follows, since $\Hmin^\eps(X|AE)_\xi=-\Hmax^\eps(X|F)_\xi=-\Hmax^\eps(X_A|F)_\psi$. 
\end{proof}

\begin{lemma}
  \label{lem:maxmin}
  For any pure state $\ket{\psi}_{AEF}$ and $0\leq \eps\leq 1$,
  \begin{align}
    &\Hmax^\eps(A|E)_\psi+\Hmin^\eps(X_A|F)_\psi\leq \log |A|\,.
  \end{align}
\end{lemma}
\begin{proof}
Here we make use of the formulation of the smooth max-entropy as a semidefinite program and appeal to the dual problem as given in~\cite{konig_operational_2009}. 
This avoids the minimax formulation inherent in \eqref{eq:smoothmax}. 
For $\psi_{ABR}$ an arbitrary purification of $\psi_{AB}$, we have
\begin{align}
\label{eq:hmaxsdp}
2^{\Hmax(A|B)_\psi}=\min\{\nu:\nu\id_B\geq Y_B,Y_{AB}\otimes \id_R\geq \psi_{ABR},Y_{AB}\geq 0\}\,.
\end{align}
Now consider $\Hmax^\eps(X|AE)_\xi$ and let $\tilde \xi_{XAE}\in\cB_\eps(\xi_{XAE})$ be such that $\Hmax^\eps(X|AE)_\psi=\Hmax(X|E)_{\tilde \xi}$.  
Next, define  $\nu$ and $\tilde Y_{XAEF}$ to be the optimal variables in \eqref{eq:hmaxsdp} so that $\nu=2^{\Hmax^\eps(X|AE)_{\xi}}$, while $\tilde Y_{AE}\leq \nu\id_{AE}$ and $\tilde Y_{XAE}\otimes\id_F\geq \tilde\xi_{XAEF}$. 
Our goal is to construct a feasible set of variables for $\Hmax^\eps(A|E)_\psi$ from $\nu$ and $\tilde Y_{XAEF}$.  

For notational simplicity, let $U$ be the isometry $U_{XA|A}$. 
Defining $Y_{AE}=U^*\tilde Y_{XAE}U$ and $\psi'_{AEF}=U^* \tilde\xi_{XAEF}U$, we have $\psi'_{AEF}\geq Y_{AE}\otimes \id_F$. 
The partial isometry $U^*$ acts as a projection $\Pi_{XA}$ and then an isometry on its support, and therefore $\psi'_{AEF}$ is a possibly subnormalized pure state and $\psi'_{AEF}\in\mathcal B_{\eps}(\psi_{AEF})$, since the purification distance only decreases under projections~\cite[Theorem 3.4]{tomamichel_framework_2012}. 
Furthermore, $U Y_{AE} U^*=\Pi_{XA}\tilde Y_{XAE}\Pi_{XA}$, which implies
\begin{subequations}
\begin{align}
Y_E 
&=\tr_{XA}[UY_{AE}U^*]\\
&=\tr_{XA}[\Pi_{XA}\tilde Y_{XAE}\Pi_{XA}]\\
&\leq \tilde Y_E\\
&\leq \nu |A|\id_{E}\,.
\end{align}
\end{subequations}
Altogether, $Y_{AE}$, $\psi'_{AEF}$, and $\lambda=\nu|A|$ are feasible for $\Hmax^\eps(A|E)_\psi$, meaning
\begin{align}
    2^{\Hmax^\eps(A|E)_\psi}\leq |A| 2^{\Hmax^\eps(X|AE)_\xi}.
  \end{align}
    The claim then follows because $\Hmax^\eps(X|AE)_\xi=-\Hmin^\eps(X_A|F)_\psi$.
\end{proof}

\section{Entropy simplification lemma}
\label{sec:simpleH}

\begin{proof}[Proof of Lemma~\ref{lem:simpleH}]
We first prove \eqref{eq:remove1}. 
Letting $\bar\sigma=\sum_y P_Y(y)\sigma_y$, note that $\psi_{YB}=\mu_Y\otimes \bar\sigma_B$, and $\psi_{XYB}=T_{XY}(\mu_X\otimes \rho_Y)T_{XY}^*$, where $T_{XY}=\sum_{xy}\ketbra{x}_X\otimes \ket{x+y}\bra{y}_Y$. 
Then, for entropies as in \eqref{eq:entropydown}, we have
\begin{subequations}
\begin{align}
\Hk(X|YB)_\psi
&=-\Dk(\psi_{XYB},\id_X\otimes \mu_Y\otimes \bar\sigma_B)\\
&=-\Dk(\mu_X\otimes \rho_{YB},\mu_X\otimes \id_Y\otimes \bar\sigma_B)\\
&=\Hk(Y|B)_\rho\,.
\end{align}
\end{subequations}
The first equality is the definition of $\Hk$, while the second follows from unitary invariance under $T^*$. 
The third equality holds by monotonicity of $\Dk$ under creating and removing $\mu_X$. 

The argument is slightly more complicated for entropies as in $\eqref{eq:entropyup}$, where we are interested in $\min_{\tau_{YB}}\Dk(\psi_{XYB},\id_X\otimes \tau_{YB})$.
Using monotonicity, we can show that without loss of generality the optimal $\tau_{YB}$ has the form $\mu_Y\otimes \tau_B$, and so the above argument can be applied to reach the desired conclusion.
Observe that $\psi_{XYB}$ is invariant under both the operation $V_{XY}=\sum_{xy}\ket{x+1}\bra{x}_X\otimes \ket{y+1}\bra{y}_Y$ as well as $U_{XY}=\sum_{xy}\omega^{x+y}\ketbra x_X\otimes \ketbra y_Y$. 
Letting $G$ be the group generated by these 
Letting $\cE(\rho_{XY})=\frac1{d^2} \sum_{j,k=0}^{d-1}U^jV^k \rho (U^jV^k)^*$, we have $\cE(\psi_{XYB})=\psi_{XYB}$ and $\cE(\id_X\otimes \tau_{YB})=\id_X\otimes \mu_Y\otimes \tau_B$. 
By monotonicity, then, 
\begin{align}
\Dk(\psi_{XYB},\id_X\otimes \mu_Y\otimes \tau_{B})\leq \Dk(\psi_{XYB},\id_X\otimes \tau_{YB})\,,
\end{align}
and therefore the optimal $\tau_{YB}$ indeed has the desired form. 

Finally, for the smooth min- and max-entropies we can again appeal to monotonicity to ensure that the optimal state in the $\eps$-ball is also classical on $X$ and $Y$ (see also \cite[Proposition 5.8]{tomamichel_framework_2012}) and is uniform on $X$ by making use of the fact that $\psi_{XYB}$ is invariant under the map which applies $T^*$, traces out $X$, creates $\mu_X$ in its place, and finally reapplies $T$.\\ 

Now for \eqref{eq:remove2}. 
Define the purifications
\begin{align}
\ket{\psi}_{XX'YY'BR}
&=\tfrac{1}{\sqrt{|X|}}\sum_{xz}\sqrt{P_Y(z)}\ket{x}_X\ket{x}_{X'}\ket{x+z}_Y\ket{x+z}_{Y'}\ket{\sigma_z}_{BR}\quad\text{and}\\
\ket{\rho}_{YY'BR}
&=\sum_y\sqrt{P_Y(y)}\ket{y}_Y\ket{y}_{Y'}\ket{\sigma_y}_{BR}\,,
\end{align}
where $\ket{\sigma_z}_{BR}$ is a purification of $\sigma_z$ for each $z$. 
By duality, we want to show $\Hk(X|X'Y'R)_\psi=\Hk(Y|Y'R)_\rho$. 
First rewrite $\ket{\psi}$ as 
\begin{align}
\ket{\psi}_{XX'YY'BR}
&=\tfrac{1}{\sqrt{|X|}}\sum_{xz}\sqrt{P_Y(z)}\ket{y-z}_X\ket{y-z}_{X'}\ket{y}_Y\ket{y}_{Y'}\ket{\sigma_z}_{BR}\,,
\end{align}
and note that the $X'Y'R$ marginal is 
\begin{align}
\psi_{X'Y'R}
&=\frac{1}{|X|}\sum_{yz}P_Y(x)\ketbra{y-z}_{X'}\otimes \ketbra{y}_{Y'}\otimes (\sigma_z)_R\,.
\end{align}
Let $U_{XX'Y'}$ be the unitary which subtracts the value in the $Y'$ register from the $X$ and $X'$ registers, and then negates the latter two. 
For $\ket{\psi'}=U\ket{\psi}$, we have
\begin{subequations}
\begin{align}
\ket{\psi'}_{XX'YY'BR}
&=\tfrac{1}{\sqrt{|X|}}\sum_{xz}\sqrt{P_Y(z)}\ket{z}_X\ket{z}_{X'}\ket{y}_Y\ket{y}_{Y'}\ket{\sigma_z}_{BR}\\
&=\ket{\Phi}_{YY'}\ket{\rho}_{XX'BR}\,,
\end{align}
\end{subequations}
where $\ket{\rho}_{XX'BR}$ is just $\ket{\rho}_{YY'BR}$ with $Y$ and $Y'$ relabelled $X$ and $X'$, respectively. 
From this expression we immediate see that $\psi_{XX'Y'R}=\mu_Y\otimes \rho_{XX'R}$. 
Moreover, $U(\id_X\otimes \psi_{X'Y'R} )U^*=\id_X\otimes \mu_{Y'}\otimes \rho_{X'R}$. 

Therefore, for entropies $\Hk$ as in \eqref{eq:entropydown},
\begin{subequations}
\begin{align}
\Hk(X|X'Y'R)_\psi
&=-\Dk(\psi_{XX'Y'R},\id_X\otimes \psi_{X'Y'R})\\
&=-\Dk(\psi'_{XX'Y'R},U(\id_X\otimes \psi_{X'Y'R} )U^*)\\
&=-\Dk(\mu_{Y'}\otimes \rho_{XX'R},\mu_{Y'}\otimes \id_X\otimes \rho_{X'R})\\
&=\Hk(Y|Y'R)_\rho\,.
\end{align}
\end{subequations}

For entropies $\Hk$ involving a marginal optimization as in \eqref{eq:entropyup}, 
we have
\begin{align}
\Hk(X|X'Y'R)_\psi
&=\max_\tau[-\Dk(\psi'_{XX'Y'R},U(\id_X\otimes \tau_{X'Y'R} )U^*)]
\end{align}
as above. 
Due to the form of $U$, $U(\id_X\otimes \tau_{X'Y'R})U^*=\id_X\otimes U' \tau_{X'Y'R}U'^*$, where $U'$ has the same action as $U$, just not applied to $X$. 
Therefore
\begin{align}
\Hk(X|X'Y'R)_\psi = \max_\tau[-\Dk(\psi'_{XX'Y'R},\id_X\otimes \tau_{X'Y'R})]\,.
\end{align}
By monotonicity we can assume the optimal $\tau_{X'Y'R}$ has the form $\mu_{Y'}\otimes \tau_{X'R}$, at which point we can follow the above derivation to complete the proof.  
\end{proof}

\end{document}